\documentclass[11pt]{amsart}
\usepackage{amssymb}
\usepackage{amsaddr}
\usepackage[english]{babel}
\usepackage[ansinew]{inputenc}
\usepackage{graphicx}
\usepackage{graphpap}
\usepackage{accents}
\usepackage{psfrag}

\newtheorem{theorem}{Theorem}[section]
\newtheorem{lemma}[theorem]{Lemma}
\theoremstyle{definition}
\newtheorem{definition}[theorem]{Definition}

\theoremstyle{remark}
\newtheorem{remark}[theorem]{Remark}
\numberwithin{equation}{section}
\theoremstyle{plain}

\newtheorem{corollary}[theorem]{Corollary}

\newtheorem{proposition}[theorem]{Proposition}

\newcounter{mnotecount}

\newcommand{\mnotex}[1]
{\protect{\stepcounter{mnotecount}}$^{\mbox{\footnotesize $\bullet$\themnotecount}}$ 
\marginpar{
\raggedright\tiny\em
$\!\!\!\!\!\!\,\bullet$\themnotecount: #1} }

\def\defi{:=}
\def\be{\begin{equation*}}
\def\en{\end{equation*}}

\def\E{\epsilon} 
\def\P{{\mathcal C}} 
\def\M{{\mathcal M}}
\def\N{{\mathcal N}}
\def\G{{\mathcal G}}

\def\cN{N}

\def\X{{\mathfrak X}}

\def\gM{g}

\def\U{U}
\def\Y{Y}
\def\sone{\hat{s}}

\def\Kn{K^{\normal}}

\def\bmell{\bm{\ell}}
\def\etaell{\bm{\eta^{(\ell)}}}
\def\etag{\bm{\eta^{\gM}}}

\def\ellc{\underaccent{\check}{\bmell}}
\def\nn{n{}^{\mbox{\tiny $(2)$}}}
\def\ll{\ell^{\mbox{\tiny $(2)$}}}

\def\ellell{(\ell \cdot \ell)}
\def\acc{\mathfrak{a}}
\def\a{{\mathcal R}}
\def\rig{L}
\def\rigging{\xi}
\def\normal{\nu}
\def\gauge{\zeta}
\def\n{n}

\def\metdata{\{\N,\gamma,\ellc,\ll\}}
\def\hypdata{\{\N,\gamma,\ellc,\ll,\Y\}}
\def\mathypdata{\{\N,\gamma,\ellc,\ll,\Y,\rho_{\ell},\bm{J} \}}

\def\A{{\mathcal A}}
\def\gup{\A^{\sharp}}

\def\Gamo{{\stackrel{\circ}{\Gamma}}}
\def\nablao{{\stackrel{\circ}{\nabla}}}
\def\nabgam{{\nabla^{(\gamma)}}}

\def\Riem{{\mbox{Riem}}}
\def\Ein{{\mbox{Ein}}}
\def\Riemo{{\stackrel{\circ}{\Riem}}}
\def\Riemoin{{\stackrel{\circ}{R}}}
\def\Vo{{\stackrel{\circ}{V}}}
\def\Ricc{{\mbox{Ric}}}
\def\Ricco{{\stackrel{\circ}{Ric}}}
\def\Riemgam{R^{(\gamma)}}

\def\gsph{g_{\mathbb{S}^{m-1}}}
\def\sph{\mathbb{S}}

\def\ellr{\ell_{\lambda}}
\def\nr{\n^{\lambda}}
\def\zetar{\zeta^{\lambda}}
\def\Yr{\Y_{\lambda}}
\def\YT{\Y_T}

\def\UU{{\mathcal U}}
\def\tilrho{\tilde{\rho}}

\def\sign{\mbox{sign}}
\def\Z{Z}
\def\ZT{\Z_T}
\def\Zr{\Z_{\lambda}}
\def\S{S}

\usepackage{geometry}
\begin{document}

\newcommand{\bm}[1]{\mbox{\boldmath $#1$}}


\def\Journal#1#2#3#4#5#6{#1, ``#2'', {\em #3} {\bf #4}, #5 (#6).}
\def\JournalPrep#1#2{#1, ``#2'',  In preparation}
\def\Monograph#1#2#3#4{#1, ``#2'', #3 (#4)}
\def\Living#1#2#3#4#5#6#7{#1, #2. Living Rev. Relativity \textbf{#3}, #4 (#5),
URL (cited on #6): 
http://www.livingreviews.org/#7}


\def\JHDE{\em J. Hyp. Diff. Eqns.}
\def\IMRN{\em Int. Math. Res. Not.}
\def\JGP{\em J. Geom. Phys.}
\def\JDG{\em J. Diff. Geom.}
\def\CQG{\em Class. Quantum Grav.}
\def\JPA{\em J. Phys. A: Math. Gen.}
\def\PRD{{\em Phys. Rev.} \bm{D}}
\def\GRG{\em Gen. Rel. Grav.}
\def\IJT{\em Int. J. Theor. Phys.}
\def\PR{\em Phys. Rev.}
\def\RMP{\em Rep. Math. Phys.}
\def\MNRAS{\em Mon. Not. Roy. Astr. Soc.}
\def\JMP{\em J. Math. Phys.}
\def\DG{\em Diff. Geom.}
\def\CMP{\em Commun. Math. Phys.}
\def\APP{\em Acta Phys. Polon.}
\def\PRL{\em Phys. Rev. Lett.}
\def\ARAA{\em Ann. Rev. Astron. Astroph.}
\def\ANP{\em Annals Phys.}
\def\AP{\em Ap. J.}
\def\APJL{\em Ap. J. Lett.}
\def\MPL{\em Mod. Phys. Lett.}
\def\PREP{\em Phys. Rep.}
\def\AASF{\em Ann. Acad. Sci. Fennicae}
\def\ZP{\em Z. Phys.}
\def\PNAS{\em Proc. Natl. Acad. Sci. USA}
\def\PLMS{\em Proc. London Math. Soth.}
\def\AIHP{\em Ann. Inst. H. Poincar\'e}
\def\ANYAS{\em Ann. N. Y. Acad. Sci.}
\def\SPJ{\em Sov. Phys. JETP}
\def\PAWBS{\em Preuss. Akad. Wiss. Berlin, Sitzber.}
\def\PPLL{\em Phys. Lett. A }
\def\QJRAS{\em Q. Jl. R. Astr. Soc.}
\def\CR{\em C.R. Acad. Sci. (Paris)}
\def\CP{\em Cahiers de Physique}
\def\NC{\em Nuovo Cimento}
\def\AM{\em Ann. Math.}
\def\APP{\em Acta Physica Polonica}
\def\BAMS{\em Bulletin Amer. Math. Soc}
\def\CPAM{\em Commun. Pure Appl. Math.}
\def\PJM{\em Pacific J. Math.}
\def\ATMP{\em Adv. Theor. Math. Phys.}
\def\PRSA{\em Proc. Roy. Soc. Lond. A.}
\def\APPT{\em Ann. Poincar\'e Phys. Theory}
\def\RPM{\em Rep. Math. Phys.}
\def\AHP{\em Annales Henri Poincar\'e}

\title{Hypersurface data: General properties and Birkhoff theorem 
in spherical symmetry}

\author{Marc Mars}

\address{
Instituto de F\'{\i}sica Fundamental y Matem\'aticas, IUFFyM\\
Universidad de Salamanca\\
Plaza de la Merced s/n \\
37008 Salamanca, Spain \\
e-mail: marc@usal.es. \\}

\begin{abstract}
  The notions of (metric) hypersurface data were introduced in \cite{MarsGRG} 
as a tool to analyze, from an abstract viewpoint,
hypersurfaces of arbitrary signature in pseudo-riemannian
manifolds. In this paper, general geometric properties of these notions are 
studied. In particular, the properties of the gauge group inherent to the
geometric construction are analyzed and the metric hypersurface connection
and its corresponding curvature tensor are studied. The results
set up the stage for various potential applications. The particular  but relevant case of spherical
symmetry is considered in detail. In particular, a collection
of gauge invariant quantities and a radial  covariant derivative
is introduced, such that the constraint equations of the Einstein field equations with matter can be written
in a very compact form. The general solution of these equations
in the vacuum case and Lorentzian
ambient signature is obtained, and a generalization of the Birkhoff
theorem to this abstract hypersurface setting is derived.

\end{abstract}



\maketitle

\section{Introduction}

The geometric character of the gravitational field requires that 
the initial value problem  for the gravitational field
is quite different than for other evolutionary systems. The
initial data does not consist simply in providing 
the initial values of the fields
(and as many initial time derivatives as needed) at a given initial
initial hypersurface. Instead, one provides an $(n-1)$- manifold and geometric
data on that manifold (usually satisfying suitable constraint equations,
depending on the specific gravity theory at hand) and asks for  the existence
of a spacetime where this data can be  embedded so that the geometry
of the embedded hypersurface suitably agrees with the prescribed geometric data. 
Sometimes this essentially geometric nature of the problem may be
somewhat obscured by the way how the evolution problem is set up.
For instance, in the case of the
Cauchy problem for the Einstein (say vacuum) field equations,
one knows a priori there exists a globally
hyperbolic Ricci flat spacetime where the abstract initial manifold
(and data) can be embedded as a spacelike hypersurface. One can then adapt coordinates, say $\{ t,x^i\}$  so that $\{ t=0\}$ corresponds
to this embedded hypersurface. By this method, the problem has been 
converted into a more standard initial value problem (although the diffeomorphism
freedom remains, and needs to be addressed e.g. by imposing coordinate
conditions such as, for instance, harmonic coordinates). In this case, however,
it is clear and well-understood that the initial data are given in a
a fully abstract manifold, completely detached from the
spacetime one intends to construct.

Another  instance of this conversion of an intrinsically
geometric initial value problem into
a standard PDE initial value problem occurs
in the Goursat (or
characteristic) initial value problem for the Einstein (vacuum) field
equations. Here the initial values are given in a pair of null
hypersurfaces intersecting on a spacelike hypersurface (or 
on the null cone of a point). One also adapts coordinates
to these hypersurfaces and solves a PDE problem. The existence of solutions
of the PDE proves at the same time the existence of the spacetime. It should be noted, however, that in the characteristic initial value problem as addressed
so far in the literature \cite{ZumHagen, Rendall, Caciotta,
Luk, Choquet, ChruscielPaetz},
one does not start with abstract geometric data and
asks whether there exists a spacetime where this data can be suitably embedded.
Instead, one assumes the hypersurface(s) to be embedded, adapts coordinates
to it, fixes values of a suitable subset of the spacetime metric
components (and of their derivatives)  on this hypersurface 
and solves the PDE problem. Although this procedure is perfectly
fine to show existence, it obscures rather strongly the essentially
geometric nature of the initial data. It also obscures the
fundamentally  abstract (in the sense of detached from the
spacetime to be constructed) nature of that data.  The fundamental reason
behind this  is the fact for null hypersurfaces (or more generally,
hypersurfaces that may have null points) 
it is much harder to define geometric structures that one can then fully
detach from the spacetime.

In \cite{MarsGRG} 
a notion of abstract hypersurface data, fully detached from the
spacetime was put forward.  The procedure followed there was a
top-bottom approach, i.e. one considered an embedded hypersurface (of fully 
unrestricted causal character) and identified 
geometric structures in the hypersurface that could then be promoted
to structures in an abstractly defined manifold (i.e. not  embedded).
The main result in \cite{MarsGRG} is, besides the introduction of geometric structures
that capture the geometry of embedded hypersurfaces at a fully detached level,
the derivation of constraint equations that this data  must
necessarily satisfy in order to be embeddable in a spacetime. 
The results were also applied to define an abstract notion of
spacetime shell, corresponding to a hypersurface where two spacetimes can
be matched so that matter and/or gravitational field are (distributionally)
supported on this codimension one submanifold. This abstract notion
of hypersurface data that can cope with null points opens up the possibility
to address several interesting initial value problems. In particular,
proving well-posedness of the characteristic initial value problem
given abstract hypersurface data, thus providing a fully
geometric initial value formulation of the problem and putting it at
the same level as the standard Cauchy problem.
Another interesting problem that may be attempted is to study well-posedness
of the abstract achronal initial value problem (i.e. when the abstract data
contains no timelike points).  It should be emphasized that the notion of
abstract hypersurface data
is useful not only to address geometric initial value problems of a more
general causal nature, but even to study the standard Cauchy initial
value problem. This is because  the notion of hypersurface data has a
built-in gauge freedom, that is a essential
to deal with null points, and provides a very flexible framework in all cases,
including the spacelike or timelike ones.  Another interesting setup where the abstract notion can be relevant is the study of so-called Killing initial data
(or KID) equations. In the spacelike Cauchy \cite{Coll, Moncrief, BeigChrusciel}
or null caracteristic \cite{ChruscielPaetz2} case, it has been possible to find necesary and sufficient conditions that the data must satisfy so that the spacetime one constructs by solving the field  equations (say vacuum) admits a Killing vector. The abstract geometric structures  provide a suitable framework to unify and extend these type of results to signature changing situations.

The fully abstract and detached construction of geometric data on the hypersurface may also find applications in formulations of General Relativy as
a constrained system (\`a la Dirac). This, and the associated
canonical formalism that it leads to, is potentially useful for an eventual theory of quantum gravity (see \cite{Friedrich}, \cite{Thiemann} and references therein).

The approach we follow here is a bottom-up approach. i.e. we define
all the notions at the fully abstract level and derive their general
properties without making any a priori assumption concerning the embeddedness
of the data in a spacetime. The connection with embedded data is made by 
quoting the results in \cite{MarsGRG}. The aim of the paper is to analyze general
properties of the geometric concept of abstract hypersurface data. 
This will provide a convenient stage to analyze the well-posedness issues 
mentioned above or, hopefully, many problems related to the abstract geometry 
of hypersurfaces.  As an example of the capabilities of the method,
we consider the simplest case of spherically symmetric hypersurface
data. After introducing suitable gauge invariant quantities, and a suitable
radial covariant derivative we pose and solve the vacuum constraint equations.
This leads to a generalized Birkhoff theorem.

The organization of the paper is as follows. In  Section \ref{Summary}
we recall the definitions of metric hypersurface data, 
hypersurface data and  describe its fundamental properties, including the
notions of null and non-null points, and a canonical volume form. The
connection with the geometry of codimension one submanifolds in ambient pseudo-riemannian spaces is made via the notion of embedded  (metric) hypersurface
data \cite{MarsGRG} which we recall. Section \ref{GaugeStruc} is devoted
to studying the inherent gauge freedom that the abstract data carries with.
In particular we identify the underlying gauge group and the effect
of a gauge transformation in the case of embedded data. We also 
determine the elements of the gauge group that leave the data invariant,
and analyze the underlying reason why non-null points behave very differently
to null points in this respect.  Sections \ref{Sect:Connection} and
\ref{Sect:Curvature} are
devoted respectively to analyzing the properties of the metric hypersurface
connection and the corresponding curvature. In particular we show existence
and uniqueness of this connection, find its behaviour under gauge
transformations and compute several identities satisfied by its 
curvature tensor. After recalling the notion of matter-hypersurface data in
Section \ref{Sect:Matter}  we devote the reminder of the paper to study in detail
the spherically symmetric case. In Section \ref{Sect:Spher} we define this notion and 
find a number of gauge covariant quantities as well as a radial
covariant  derivative that allows one to write down the constraint equations
in a very compact form. In Section \ref{Sect:Birkhoff} we solve the constraint
equations in the vacuum case and Lorentz ambient signature and prove an
embedding results that generalized the Birkhoff theorem to the abstract
hypersurface data level.

\section{Abstract data and embedded data}

\label{Summary}

The two fundamental objects studied in this paper
are the so-called metric hypersurface data and hypersurface data. These
notions were introduced in \cite{MarsGRG} in the process of extracting
the fundamental ingredients that allow one to describe 
the geometry of hypersurfaces  of pseudo-riemannian
manifolds in a fully
abstract way, i.e. without viewing them as embedded in any ambient space.

\begin{definition}[Metric hypersurface data and hypersurface data]
A four-tuple
$\{\N^m,\gamma,\ellc,\ll\}$ where
$\N$ is a smooth $m$-dimensional manifold\footnote{Manifolds are assumed to
  be connected in this paper.}, 
$\gamma$ is a symmetric two-covariant tensor field,
$\ellc$  a one-form field and  $\ll$ a scalar field, defines
{\bf metric hypersurface data} provided the
symmetric $2$-covariant tensor $\A|_p$ on $T_p \N \times \mathbb{R}$ defined by
\be
\A|_p((W,a),(Z,b))\defi \gamma|_p(W,Z) + a \, \ellc|_p(Z)  + b \,
\ellc|_p(W) + a \, b \, \ll|_p
\en
is non-degenerate at every $p \in \N$. 
A five-tuple  $\{\N,\gamma,\ellc,\ll, {\bf Y}\}$ where
$\{\N,\gamma,\ellc,\ll\}$ is metric hypersurface data and
${\bf Y}$ is a symmetric two-covariant tensor field is called
{\bf hypersurface data}. 
\end{definition}

All tensor field in this paper are assumed to be smooth. Obviously these
notions also make sense under finite, sufficiently high,
differentiability assumptions.

Given metric hypersurface data the following geometric constructions
can be made. Since $\A|_p$ is a symmetric two-covariant
and non-degenerate tensor on $T_p \N \oplus \mathbb{R}$, there
exists a (unique) symmetric, two-contravariant tensor
$\gup|_p$ defined by the property $\gup|_p (\cdot, \A|_p (V, \cdot ))
= V$ for any vector $V \in T_p \N \oplus \mathbb{R}$. From $\gup|_p$ we can define a symmetric contravariant tensor $P|_p$, vector $\n|_p$ and scalar
$\nn|_p$ on $T_p \N$ by the decomposition
\begin{align}
\gup|_p ((\bm{\alpha},a), (\bm{\beta},b)) =
P|_p (\bm{\alpha},\bm{\beta}) + a \, \n|_p (\bm{\beta}) +
b \, \n |_p (\bm{\alpha}) + a \, b \, \nn|_p, \quad \quad
\bm{\alpha}, \bm{\beta} \in T^{\star}_p \N, \quad a,b \in \mathbb{R}.
\label{inverse}
\end{align}
We also define $\A, \A^{\sharp}$, $P$, $\n$ and
$\nn$ to be the corresponding tensor fields. The defining condition for
$\gup$ is equivalent to
\begin{align}
&\gamma(\n,\cdot ) +  \nn \ellc = 0 \label{EqP1bis} \\
& \ellc(\n) = 1-  \nn \ll \label{EqP2bis} \\
& P (\cdot, \ellc) = - \ll \n, \label{EqP3bis} \\
& P ( \cdot, \gamma(\cdot, X )) = X - \ellc(X) \n \quad
\quad & \forall X \in T_p \N, \label{EqP4bis}  \\
& \gamma ( \cdot, P(\cdot, \bm{\alpha} )) = \bm{\alpha} - 
\bm{\alpha} (\n) \ellc  \quad
\quad &  \forall \bm{\alpha} \in T^{\star}_p \N,  \label{EqP5bis}  
\end{align}
which define uniquely $P$, $\n$ and $\nn$. In abstract index
notation, $\A$ and $\A^{\sharp}$ can be written as square $(m+1)$-matrices
as
\begin{eqnarray}
\A = 
\left ( \begin{array}{cc}
              \gamma_{ab} & \ellc_b \\
              \ellc_a & \ll 
             \end{array}
\right ),  
\quad \quad \quad \quad 
\A^{\sharp} = \left ( \begin{array}{cc}
              P^{ab} & n^b \\
              n^a & \nn 
             \end{array}
\right ),
\label{matrixA}
\end{eqnarray}
and the defining conditions
(\ref{EqP1bis})-(\ref{EqP5bis}) 
read respectively (the last two are condensed
into one when indices are used)
\begin{eqnarray}
\gamma_{ab} n^b  + \nn \ell_a = 0, \label{EqP4} \\
n^a \ell_a  + \nn \ll = 1,   \label{EqP3} \\ 
P^{ab} \ell_b + \ll n^a = 0, \label{EqP2} \\
P^{ab} \gamma_{bc} + n^a \ell_c = \delta^a_c. \label{EqP1} 
\end{eqnarray}

We emphasize that no assumption on the signature of the tensor
$\A$ is made, so unless otherwise stated the results below
are valid for any such signature.
An interesting property is that $\gamma$ has at most one degeneration
direction at each point $p \in \N$.
\begin{lemma}
Let $\{\N^m,\gamma,\ellc,\ll\}$ be metric hypersurface data and $p \in \N$.
Then the radical of $\gamma$, with the usual definition
 $\mbox{Rad}|_p \defi \{ X \in T_p  \N \, ; \,
\gamma(X,\cdot)=0 \} \subset
T_p \N$, is either zero or one-dimensional.
\end{lemma}

\begin{proof}
It suffices to show that when the radical contains a non-zero element,
then it is
one-dimensional. Assume $X_1,X_2$ are two non-zero degeneration vectors 
of $\gamma$ and define
(clearly non-zero) vectors on $T_p \N \oplus \mathbb{R}$ by
$$(
X_i, 0), \quad i = 1,2.
$$
Applying $\A|_p$  gives, respectively, the one-forms
$$
\A|_p(X_1,0) = \ellc(X_1) \mathfrak{i}|_p, \quad \quad
\A|_p(X_2,0) = \ellc(X_2) \mathfrak{i}|_p
$$ 
where $\mathfrak{i}|_p \in \left ( T_p \N \oplus \mathbb{R} \right )^{\star}$
is defined by $\mathfrak{i} (X,a) = a$.
Now, $\ellc(X_2)$ cannot vanish because $\A|_p$ is non-degenerate. Consequently
$(X_1 - \frac{\ellc(X_1)}{\ellc(X_2)} X_2, 0)$ lies in the kernel
of $\A|_p$ and hence vanishes identically, which shows that
$\{ X_1, X_2\}$ are linearly dependent and the dimension of the radical 
cannot be two or more.
\end{proof}

The previous lemma suggests the following definition.
\begin{definition}
Let $\{\N,\gamma,\ellc,\ll\}$ be metric hypersurface data and
$p \in \N$. The point $p$ is called a {\bf non-null point} 
if $\mbox{Rad}|_p = \{ 0 \}$. $p$ is a 
 {\bf null point} 
if $\mbox{dim} (\mbox{Rad}|_p) = 1$. 
\end{definition}

\begin{remark}
From the expression of an inverse matrix in terms of its minors, it follows
immediately from (\ref{matrixA}) that 
$\nn=0$ if and only if
$\mbox{det}(\gamma)=0$, i.e. iff $\gamma$ is degenerate.
So $p$ is a null point if and only if $\nn|_p =0$.
Moreover, at null points $n |_p \neq 0$ (as otherwise $\A^{\sharp}$ would be 
degenerate) and (\ref{EqP4}) shows that 
$\mbox{Rad}|_p = \langle n|_p \rangle$.
Thus, at null points $n$ defines the only degeneration direction
of $\gamma$. 
\end{remark}

\begin{definition}
The metric hypersurface data $\metdata$ is called {\bf null metric
hypersurface data} (or simply null data), if all points $p \in \N$
are null points.
\end{definition}

At non-null points, the fields $P$, $n$ and $\nn$ admit explicit expressions
in terms of the metric hypersurface data.
\begin{lemma}
\label{Pnnn_non-null}
Let $\{\N,\gamma,\ellc,\ll\}$ be metric hypersurface data and $p \in \N$
a non-null point. Then $\mbox{det}(\gamma)|_p \neq 0$ and
$(\ll - \gamma^{\sharp} (\ellc,\ellc) )|_p \neq 0$, where
$\gamma^{\sharp}|_p$ is the contravariant metric associated to
$\gamma|_p$. The tensors $P|_p$, $n |_p$ and $\nn|_p$ read (evaluating
everything at $p$)
\begin{align*}
P  = \gamma^{\sharp} + \frac{1}{\ll - \gamma^{\sharp}(\ellc,\ellc)}
\ell^{\sharp} \otimes \ell^{\sharp}, 
\quad \quad
n = - \frac{1}{\ll - \gamma^{\sharp}(\ellc,\ellc)} \ell^{\sharp},
\quad \quad 
\nn = \frac{1}{\ll - \gamma^{\sharp}(\ellc,\ellc)}
\end{align*}
where $\ell^{\sharp} := \gamma^{\sharp}(\ellc,\cdot)$.
\end{lemma}

\begin{proof}
By definition of non-null point, $\gamma|_p$ is non-degenerate, hence
a metric, and  admits an associated contravariant metric 
$\gamma^{\sharp}$. We work at $p$ from now on 
and  simplify the notation by dropping $|_p$. 
Known properties of matrices (see e.g. \cite{Schur})
state that the  determinant of $\A$  is
\begin{align}
\mbox{det} (\A)|_p = \left ( \mbox{det} \gamma \right ) 
\left ( \ll - \gamma^{\sharp}(\ellc,\ellc) \right ),
\label{reldets}
\end{align}
which in particular implies  
$\left ( \ll - \gamma^{\sharp}(\ellc,\ellc) \right ) \neq 0$ at $p$.
The explicit form of $P$, $n$, and $\nn$ follow either from using 
known forms for the 
inverse of block diagonal metrics (e.g. \cite{Schur})
or simply by direct substitution in the defining expressions 
(\ref{EqP4})-(\ref{EqP1}).
\end{proof}

It is also of interest to relate the signatures of $\gamma$ and
of $\A$ at non-null-points. To be precise, by signature
of a non-degenerate quadratic form $q$, denoted
$\mbox{sign} (q)$, we mean
the collection of $\{-1,$ $\cdots,$ $-1,$ $+1,$ $\cdots,$ $+1\}$
in the canonical
expression of the quadratic form. We view this collection as a set, hence the order of
elements is irrelevant.
\begin{lemma}
\label{signa}
Let $\metdata$ be metric hypersurface data and $p \in \N$ a non-null
point. Then the signatures of $\gamma|_p$ and $\A|_p$ are related by
\begin{align}
\mbox{sign} (\A|_p) = \mbox{sign} (\gamma|_p) \cup \{ \mbox{sign} (\ll -
\gamma(\ell^{\sharp},\ell^{\sharp})) \}.
\label{signatures}
\end{align}
\end{lemma}

\begin{proof}
Let $\{ e_{a} \}$  be a basis of $T_p \N$ in which $\gamma_p$ takes
its canonical form, i.e.$\gamma_{ab} := \gamma|_p(e_a,e_b) = \eta_{ab}$
with $\eta_{ab} = \mbox{diag} (-1, \cdots, -1, +1, \cdots +1)$.
We view $T_p \N$ as a vector subspace of $T_p \N \times \mathbb{R}$ and
define $e_{m+1} = (0,1) \subset T_p \times \mathbb{R}$. Define
$\ell_a := \A|_p (e_a,e_{m+1})$ and $\ell^a := \eta^{ab} \ell_b$ where
$\eta^{ab} := \eta_{ab}$. In the transformed basis $\{ e_a':= e_a, e_{m+1}'
:= (- \ell^b e_b,1) \}$, the components of $\A|_p$ and $\gamma |_p$ are
(note that $e_{m+1}' = e_{m+1} - \ell^b e_b$)
\begin{align*}
\A|_p (e'_a, e'_b) & = 
\A|_p (e_a, e_b) = 
\gamma|_p (e_a,e_b) = \eta_{ab}, \\
\A |_p (e'_a, e'_{m+1} ) & = \A|_p ( e_a, e_{m+1} - \ell^b e_b )
) = \ell_a - \ell^{b} \eta_{ab} = 0, \\
\A|_p (e'_{m+1},e'_{m+1} ) & = 
\A|_p (e'_{m+1}, e_{m+1} - \ell^b e'_b) =
\A|_p (e'_{m+1}, e_{m+1} ) 
= \A|_p (e_{m+1} - \ell^a e_b, e_{m+1}) \\
& = \ll - \ell^a \ell_a
= \ll - \gamma(\ell^{\sharp},\ell^{\sharp}) := \epsilon a^2, \quad
\quad \epsilon=\pm 1, a >0.
\end{align*}
Hence $\{e'_a, a^{-1}
e'_{m+1} \}$ is a canonical basis
of $\A|_p$, and  (\ref{signatures}) follows.
\end{proof}

It is remarkable that hypersurface data is sufficient to define a volume form
on the manifold.

\begin{proposition}[Volume form]
\label{volume}
Let $\{\N,\gamma,\ellc,\ll\}$ be metric hypersurface data 
and assume
$\N$ to be oriented. Then the following local expression 
in any  chart $\{x^a \}$ defines a volume form on $\N$,
\begin{align}
\label{volformdef}
\etaell = \sigma \sqrt{| \mbox{det} (\A ) |} \, \E
\end{align}
where $\sigma =+1 (-1)$ if the chart is positively (negatively) oriented, $\E$ is the Levi-Civita totally antisymmetric
symbol and $| \cdot |$ is the absolute value.
\end{proposition}

\begin{remark}
  Note that at non-null points this volume form is different from the
  metric volume form $\bm{\eta_{\gamma}} := \sigma \sqrt{| \mbox{det} (\gamma)|}
  \epsilon$. The relation between both follows directly from (\ref{reldets})
 and reads
  \begin{align*}
    \etaell =  \sqrt{|\ll - \gamma^{\sharp}(\ellc,\ellc)|} \, \bm{\eta_{\gamma}}
    = \frac{1}{\sqrt{|\nn|}} \bm{\eta_{\gamma}},
  \end{align*}
  where in the second equality we used Lemma \ref{Pnnn_non-null}.
  \end{remark}

  \begin{proof}
We only need to check how  $| \mbox{det} (\A)|$ changes under
positively oriented changes of coordinates.
Fix a point $p \in \N$ and a positively oriented
coordinate chart $x$ defined on a neighbourhood
$U$ of $p$. Let $x'$ be another positively oriented chart on $U$ and $x'(x)$
be the coordinate change. Define
$J^{a'}_{\,\,b} = \left . \frac{\partial x'{}^{a'}}{\partial x^b} 
\right |_{x(p)}$ be the Jacobian
at $p$ and $J$ its (positive) determinant.  Since $\ell_{a} = J^{a'}_{\,\,a}
\ell'_{a'}$
and $\gamma_{ab} = J^{a'}_{\,\,a} J^{b'}_{\,\,b} \gamma'_{a'b'}$, with obvious 
notation, we have ($T$ denotes transpose)
\be
\mbox{det} (\A) = \left | \left ( \begin{array}{ll}
                                   (\gamma) & (\ell)^T \\
                                   (\ell) & \ll
                                   \end{array}
\right ) \right | =
\left | \left ( \begin{array}{ll}
                (J) & 0^T \\
                0 & 1
                \end{array} 
\right ) \left ( \begin{array}{ll}
                 (\gamma^{\prime}) & (\ell^{\prime})^T \\
                 (\ell^{\prime})  & \ll{}^{\prime}
                 \end{array}
\right ) \left ( \begin{array}{ll}
                (J^{T}) & 0^T \\
                0 &  1
                \end{array} 
\right ) \right | = 
\mbox{det} (J)^2 \mbox{det} (\A^{\prime}),
\en
where $(\ell)$ is a row vector with entries $\ell_a$ and similarly
for the other objects. In particular $(J)$ is the matrix
with entries $(J)^{a'}_{a} = J^{a'}_a$
where the upper index denotes column and the lower index row.
\end{proof}

The definition of metric hypersurface data and hypersurface data makes no 
reference to any ambient space.  The connection 
of the abstract definition 
with the geometry of embedded hypersurfaces in pseudo-riemannian
manifolds
is provided by the following Definition 
presented in \cite{MarsGRG} and the corresponding embedding
Proposition \ref{propembedded}.  The notion of rigging in the context
of codimension one submanifolds dates back to Schouten \cite{Schouten}:

\begin{definition}
A metric hypersurface data $\{\N^m,\gamma,\ellc,\ll\}$ is
{\bf embedded} in a pseudo-riemannian manifold $(\M^{m+1},\gM)$ (called
ambient space) if there
exists an embedding $\Phi : \N \rightarrow \M$ and 
a vector field $\rigging$ along $\Phi(\N)$ everywhere transversal to $\Phi(\N)$ (a 
so-called {\bf rigging})
such that,
\begin{equation}
\Phi^{\star} (\gM) = \gamma, \quad \Phi^{\star} \left ( \gM(\rigging, \cdot) \right ) =
\ellc,
\quad   \Phi^{\star} (\gM(\rigging,\rigging) ) = \ll. 
\label{condembedded}
\end{equation} 
The hypersurface data 
$\{\N,\gamma,\ellc,\ll,Y\}$ 
 is {\bf embedded} if
$\{\N,\gamma,\ellc,\ll\}$  is embedded and, in addition
\begin{equation}
\frac{1}{2} \Phi^{\star} \left ( {\pounds}_{\rigging\,\,} \gM \right ) = {\bf Y}.
\label{EmbHyp}
\end{equation}
\end{definition}

\begin{remark}
To simplify the notation, we will often use the same symbol
to denote a scalar function in the ambient space
and its pull-back under $\Phi$.
Thus, the third condition in (\ref{condembedded}) may also be written as
$\gM(\rigging,\rigging)  = \ll$.
\end{remark}

\begin{remark}
It is clear from the definition that a necessary condition for given metric
data to be embedded in $(\M,g)$ is that the signature of $\A$ is the same
as the signature of $g$. For this reason, we denote the
signature of $\A$ as {\bf ambient
signature} (irrespectively of whether the data is embedded or not).
\end{remark}

\begin{remark}
Note that $\Phi^{\star} (\pounds_{\rigging} g)$ requires an
extension of $\rigging$ off $\Phi(\N)$ but the result is independent
of this extension (it only depends on $\rigging$ along $\Phi(\N)$).
\end{remark}

\begin{remark}
\label{orientab}
A necessary and sufficient condition 
for an embedded hypersurface $\Phi(\N)$ to admit
a smooth transversal vector field $\rigging$ is that
$\Phi(\N)$  admits a smooth nowhere zero field
of normal one-forms (see Lemma 1 in \cite{MarsGRG}). In turn,
assuming $(\M,g)$ to be orientable
the existence of such smooth normal field is equivalent
to the orientability of $\N$ (see e.g. \cite{Guillemin} p. 106).
\end{remark}

\begin{proposition}
[\cite{MarsGRG}, cf. also \cite{MarsSenovilla1993}] 
\label{propembedded}
Let $\{\N,\gamma,\ellc,\ll, {\bf Y}\}$ be
embedded hypersurface data with embedding
$\Phi$, ambient space $(\M,\gM)$ and rigging vector $\rigging$.
Let $\bm{\normal}$ be the (unique) normal-one form of $\Phi(\N)$
satisfying $\bm{\normal}(\rigging) =1$. Consider a (local) basis $\{ e_a \}$ of 
$T \N$
and define
$\hat{e}_a := \Phi_{\star} (e_a)$.   
Then $\{ \rigging, \hat{e}_a \}$ is a (local)
basis of
$T \M|_{\Phi(\N)}$
and we can consider its dual 
basis  $\{ \bm{\normal}, \bm{\hat{\omega}}^a\}$. Then, the following 
decompositions
hold
\begin{align}
\normal \defi g^{\sharp} (\bm{\normal},\cdot) &=  n^a \hat{e}_a + \nn 
\rigging, \label{n}\\
g(\rigging,\cdot) & = \ell_a \bm{\hat{\omega}^a}  
+ \ll \bm{\normal}, \label{bmell} \\
 g(\hat{e}_a, \cdot) & = \gamma_{ab}  \bm{\hat{\omega}}^b
 + \ell_a \bm{\normal}, \label{bme} \\
g^{\sharp}(\bm{\omega}^a, \cdot) &= P^{ab} \hat{e}_b + n^a \rigging. \label{omegaa}
\end{align}
\end{proposition}

\begin{remark}
\label{nn}
This proposition implies, in particular, that
$\gamma$ is the first fundamental form of 
$\Phi(\N)$. It also justifies the superindex $(2)$
in the notation $\nn$ because, multiplying 
(\ref{n}) by
$\bm{\normal}$ and using $\bm{\normal} (\rigging ) = 1$ yields
\begin{align*}
g^{\sharp}(\bm{\normal}, \bm{\normal}) = \nn
\end{align*}
The justification for the $(2)$ in $\ll$ comes from the last in 
(\ref{condembedded}).
\end{remark}

The volume form $\etaell$  acquires
a clear geometric interpretation for embedded hypersurface
data.
\begin{lemma}
Let $\{\N,\gamma,\ellc,\ll\}$ be
embedded metric hypersurface data with embedding
$\Phi$, ambient space $(\M,\gM)$ and rigging vector $\rigging$. Assume that 
$\M$ is oriented with metric volume form $\etag$. Then $\N$ is orientable
and, after choosing an orientation, it holds
\be
\etaell(X_1,\cdots, X_m) = \hat{\sigma} \etag(\rigging,\Phi_{\star}(X_1),
\cdots, \Phi_{\star}(X_m)), \quad \quad X_i \in \X(\N)
\en
where  $\hat{\sigma}$ is the product of the orientations of $\{\rigging, 
\Phi_{\star}(X_1), \cdots,
\Phi_{\star}(X_m)\}$ and  $\{X_1, \cdots, X_m\}$.
\end{lemma}

\begin{proof}
The orientability of $\N$ is necessary for the embeddability of the metric
hypersurface data in an oriented
ambient space (see Remark \ref{orientab}). It suffices to prove the equality
for a positively oriented basis $\{e_1, \cdots, e_m\}$ of $T_p \N$. 
Then $\{ \rigging|_{\Phi(p)}, \Phi_{\star}(e_1), \cdots, \Phi_{\star}(e_m) \}$ is a basis of
$T_{\Phi(p)} \M$. Let $\hat{\sigma}$ be its orientation with respect to 
$g$. From the definition of metric volume form
\begin{align*}
\etag(\xi|_{\Phi(p)}, \Phi_{\star}(e_1), \cdots, \Phi_{\star}(e_m) \} = \hat{\sigma} \sqrt{|\mbox{det} (g)|},
\end{align*}
where $(g)$ is the symmetric matrix  with entries $g (1,1):= g|_{\Phi(p)} (\xi|_{\Phi(p)}, \xi|_{\Phi(p)}) = \ll |_p$, $g(1,a+1):= g |_{\Phi(p)} (\xi|_{\Phi(p)},\Phi_{\star}(e_i))
= \ell_a |_p$ and
$g(a+1,b+1):= g|_{\Phi(p)} (\Phi_{\star}(e_a),\Phi_{\star}(e_b)) 
= \gamma_{ab} |_p$, where we used the definition of embedded metric
hypersurface data and $\ell_a |_p$, $\gamma_{ab}|_p$ are the components
of $\ellc|_p$ and $\gamma|_p$ in the basis $\{ e_a \}$.
Since the matrix $\A|_p$ in the basis $\{ e_a\}$  has the same entries as $(g)$ 
 we conclude
\begin{align*}
\etag(\xi|_{\Phi(p)}, \Phi_{\star}(e_1), \cdots, \Phi_{\star}(e_m) ) = \hat{\sigma} 
\sqrt{|\mbox{det} \A |_p} = \hat{\sigma} \etaell(e_1, \cdots, e_m).
\end{align*}
\end{proof}

\begin{remark}
A useful  quantity that can be defined in any
hypersurface data $\hypdata$ is the 
symmetric tensor
\begin{align*}
\Kn \defi \nn \Y + \frac{1}{2} \pounds_{n} \gamma + \ellc \otimes_s d \nn
\end{align*}
where $\otimes_s$ denotes symmetrized tensor product, i.e.
$\alpha \otimes_s \beta \defi
\frac{1}{2} ( \alpha \otimes \beta + \beta \otimes \alpha)$.
When the data 
is embedded, $\Kn$ is the second fundamental
form of $\Phi(\N)$  with respect to the normal
$\bm{\normal}$. This follows easily from
(\ref{EmbHyp})  
and (\ref{n})
(see \cite{MarsGRG} for details).

\end{remark}

\section{Gauge structure}
\label{GaugeStruc}
An important property of hypersurface data is that it has a built-in gauge
freedom, connected to the fact that for embedded hypersurface data the rigging
vector is non-unique. The relevant definition is as follows 
(we correct a typo
in \cite{MarsGRG} concerning the definition of the gauge function $u$).
We denote by ${\mathcal F}^{\star}(\N)$  the set of smooth nowhere
zero real functions on $\N$ and $\X(\N)$ refers, as usual, to the set
of (smooth) vector fields on $\N$.
\begin{definition}
\label{gauge}
Let $\hypdata$ be hypersurface data.
Let $u \in {\mathcal F}^{\star}(\N)$ and $\gauge \in \X(\N)$.
The {\bf gauge transformed} hypersurface data with gauge parameters $(u,\gauge)$ 
is defined as 
\begin{align}
\G_{(u,\gauge)}(\gamma)  & := \gamma, \quad 
\G_{(u,\gauge)} (\ellc) := 
u \left ( \ellc + \gamma(\gauge, \cdot) \right ), \quad
\G_{(u,\gauge)} (\ll) =: u^2 \left ( \ll + 2 \ellc (\gauge) 
+ \gamma(\gauge,\gauge)  \right ),
\label{gaugetrans} \\
\G_{(u,\gauge)} (Y) & := u Y + \ellc \otimes_s du 
+  \frac{1}{2} \pounds_{u \gauge } \gamma.
\nonumber
\end{align}
\end{definition}
\begin{remark}
We will sometimes apply gauge transformations to 
just metric hypersurface data. The definition is the same but
ignoring the $\Y$ tensor.
\end{remark}

The following lemma shows that the name ``gauge transformations'' is
justified, as the collection of transformation forms a group.

\begin{lemma}
The set of gauge transformations forms a group $\G = {\mathcal F}^{\star}(\N)
\times \X(\N)$ with composition law
\begin{align}
(u_2,\gauge_2) \circ (u_1, \gauge_1) = \left ( u_2 u_1, \gauge_1 + u_1^{-1}\gauge_2
\right ). \label{comp}
\end{align}
\end{lemma}

\begin{proof}
The composition law in $\G$ is defined by the requirement
$\G_{(u_2,\gauge_2) \circ (u_1,\gauge_1)} = \G_{(u_2,\gauge_2)} \circ
\G_{(u_1,\gauge_1)}$, provided this condition defines
a unique element  $(u_2,\gauge_2) \circ (u_1,\gauge_1)
\in {\mathcal F}^{\star}(\N) \times \X(\N)$. We compute the composition
of transformations when acting on $\ellc$:
\begin{align*}
\G_{(u_2,\gauge_2)} (\G_{(u_1,\gauge_1)} (\ellc)) & =
                                                    u_2 \left ( \G_{(u_1,\gauge_1)} (\ellc) + \gamma(\gauge_2, \cdot) \right )
                                                    = u_2 \left ( u_1
                                                    \left ( \ellc + \gamma(\gauge_1, \cdot ) \right ) + \gamma(\gauge_2, \cdot ) \right ) \\
                                                  &  = u_2 u_1 \left ( \ellc + \gamma( \gauge_1 + u_1^{-1} \gauge_2, \cdot ) \right )
\end{align*}
which shows that necessarily  
$(u_2,\gauge_2) \circ (u_1, \gauge_1) = \left ( u_2 u_1, \gauge_1 + u_1^{-1} \gauge_2
\right ) := (u_3,\gauge_3)$. It remains to show that the
equality $G_{(u_3,\gauge_3)} = \G_{(u_2,\gauge_2)} \circ \G_{(u_1,\gauge_1)}$ still holds
when applied to $\gamma, \ll$ and $Y$. $\gamma$ is 
gauge invariant so there is nothing to be checked. Acting on $\ll$:
\begin{align*}
  \G_{(u_2,\gauge_2)} \circ & \, \G_{(u_1,\gauge_1)} (\ll)  =
                              u_2^2 \left ( \G_{(u_1,\gauge_1)} (\ll)   + 2 \G_{(u_1,\gauge_1)} (\ellc) (\gauge_2) + \gamma(\gauge_2,\gauge_2 ) \right )   
                                \\
& = u_2^2 \left ( u_1^2 \left ( \ll + 2 \ellc (\gauge_1)  
+ \gamma(\gauge_1,\gauge_1 ) \right ) 
+ 2 u_1 ( \ellc + \gamma(\gauge_1, \cdot )) (\gauge_2) 
+ \gamma(\gauge_2,\gauge_2) \right ) \\
& = ( u_1 u_2)^2 \left ( \ll +
2 \ellc (\gauge_1 + u_1^{-1} \gauge_2 ) 
+ \gamma ( \gauge_1 + u_1^{-1} \gauge_2, \gauge_1 + u_1^{-1} \gauge_2) \right )=
\G_{(u_3,\gauge_3)} (\ll).
\end{align*}
Concerning $Y$,
\begin{align*}
  \G_{(u_2,\gauge_2)} \circ & \, \G_{(u_1,\gauge_1)} (Y)  =
                              u_2 \G_{(u_1,\gauge_1)} (Y) +
                              \G_{(u_1,\gauge_1)} (\ellc)  \otimes_s du_2 
                              +\frac{1}{2} \pounds_{u_2 \gauge_2} \gamma \\
& =
u_2 \left ( u_1 Y + \ellc \otimes_s du_1 + \frac{1}{2} \pounds_{u_1 \gauge_1} \gamma
\right ) + u_1 \left ( \ellc + \gamma(\gauge_1,\cdot) \right ) \otimes_s du_2
+\frac{1}{2} \pounds_{u_2 \gauge_2} \gamma \\
& = u_1 u_2  Y + 
\ellc \otimes_s \left ( u_2 d u_1 + u_1 d u_2 \right )
+ \frac{1}{2} \big ( u_2 \pounds_{u_1 \gauge_1}  + \pounds_{u_2 \gauge_2} \big )
\gamma + \gamma( u_1 \gauge_1, \cdot ) \otimes_s du_2.
 \end{align*}
We now use the standard (and easy to prove) 
identity on symmetric two-covariant
tensors $T$,
$f \pounds_{v} T =
\pounds_{f v} T - 2 T(v,\cdot) \otimes_s df$
and find
\begin{align*}
\G_{(u_2,\gauge_2)} \circ \G_{(u_1,\gauge_1)} (Y)  & =
u_1 u_2 Y + \ellc \otimes_s d (u_1 u_2) 
+ \frac{1}{2} \pounds_{ u_1 u_2 \gauge_1 + u_2 \gauge_2} \gamma \\
& = u_3 Y + \ellc \otimes_s d u_3 
+ \frac{1}{2} \pounds_{ u_3 \gauge_3 } \gamma
= \G_{(u_3,\gauge_3)} (Y).
\end{align*}
The fact that the composition law (\ref{comp}) is a group operation
is immediate. In particular, the inverse of $(u,\gauge)$ is given by
\begin{align*}
  (u,\gauge)^{-1} = (u^{-1}, -u \gauge) 
\end{align*}
and the identity element is $e = (1,0)$.
\end{proof}

\begin{remark}
  At any given point $p \in \N$, the gauge transformation restricts to the finite dimensional Lie group $\G_p := \mathbb{R}^{+} \times T_p \N$, with group law given by  (\ref{comp}). This Lie group is isomorphic (with the inversion isomorphism $g \rightarrow g^{-1}$) to the semidirect product $\mathbb{R}^+ \ltimes_{\varphi}T_p \N$ (each factor endowed with the standard group operation) where the homomorphism $\varphi : \mathbb{R}^+
  \rightarrow \mbox{Aut} (T_p \N)$ is given by $\varphi(u) (\gauge) =  u^{-1} \gauge$.
  It is worth writing  down a basis of left-invariant vector fields on $\G_p$ and compute the corresponding Lie algebra. The result is
  \begin{align*}
   & X^L_u = u \partial_u - \zeta^b \partial_{\zeta^b}, \quad \quad
    X^L_{\zeta^a} = \partial_{\zeta^a} \\
    & \left  [X^L_u , X^L_{\zeta^a} \right ] = X^L_{\zeta^a},  
    \quad \quad [ X^L_u, X^L_u ] =  [X^L_{\zeta^a}, X^L_{\zeta^b} ] =0.
  \end{align*}
\end{remark}

Our next result gives the connection between 
gauge transformations and the
freedom in the choice of rigging in the case that 
metric hypersurface
data is embedded. The effect of change of rigging
on embedded  data was obtained in  \cite{MarsGRG}
and, indeed motivated the abstract definition of gauge transformation. The
following result shows the converse, namely that
under an arbitrary  change of gauge, embedded hypersurface data remains
embedded.
\begin{proposition}
\label{changerig}
Let $\hypdata$ be  hypersurface data with embedding
$\Phi$, ambient space $(\M,g)$ and rigging vector $\rigging$. For 
any gauge parameters $(u,\gauge)$,  
$\{ \N$, $\G_{(u,\gauge)} (\gamma)$, $\G_{(u,\gauge)}(\ellc)$, $\G_{(u,\gauge)}(\ll)$,
$\G_{(u,\gauge)}(\Y)\}$
is embedded hypersurface data, with the same ambient space
$(\M,g)$, same embedding $\Phi$, and rigging $\G_{(u,\gauge)} (\rigging)$ given
by
\begin{align*}
\G_{(u,\gauge)} (\rigging) 
\defi u \left ( \rigging + \Phi_{\star} (\gauge) \right ).
\end{align*}
\end{proposition}


\begin{proof}
For simplicity we write $\rigging^{\, \prime} := \G_{(u,\gauge)} (\rigging)$. 
From the definition of embedded hypersurface data we need to check 
that 
\begin{align}
\Phi^{\star}(g (\rigging^{\, \prime}, \cdot)) 
& = u \left ( \ellc + \gamma(\gauge,\cdot) \right ), \label{id1} \\ 
g \left ( \rigging^{\, \prime},  \rigging^{\, \prime} \right ) 
& = u^2 \left ( \ll + 2 \ellc(\gauge) + \gamma(\gauge,\gauge) \right ),
\label{id2} \\
\Phi^{\star} \left (  \frac{1}{2} 
\pounds_{\G_{(u,\gauge)}(\rigging)} g 
\right )  & = \G_{(u,\gauge)} (\Y).
\label{id3}
\end{align}
Since 
\begin{align*}
g ( \rigging^{\, \prime}, \cdot )
= g( \G_{(u,\gauge)} (\rigging), \cdot) 
= g ( u \left ( \rigging + \Phi_{\star} (\gauge) \right ), \cdot )
= u g (\rigging, \cdot) + u g ( \Phi_{\star} (\gauge), \cdot),
\end{align*}
its pull-back is, using that $\{ \N,\gamma,\ellc,\ll \}$ is embedded with
rigging $\rigging$,
\begin{align*}
\Phi^{\star} ( g(\rigging^{\, \prime}, \cdot) )
 = u \ellc + u \Phi^{\star} ( g (\Phi_{\star} (\gauge), \cdot )
= u \ellc + u \Phi^{\star}(g) \left ( \gauge, \cdot \right ) =
u \ellc + u \gamma ( \gauge, \cdot ) 
\end{align*}
which establishes (\ref{id1}). On the other hand
\begin{align*}
g ( \rigging^{\, \prime}, \rigging^{\, \prime})  =
u^2 \left ( g (\rigging, \rigging) +
2 g(\rigging,\Phi_{\star} (\gauge)) + 
g(\Phi_{\star}(\gauge),\Phi_{\star}(\gauge)) \right ).
\end{align*}
Taking the pull-back yields (\ref{id2}) immediately. Finally 
(\ref{id3}) is checked as follows
\begin{align*}
\frac{1}{2}
\Phi^{\star} \left (  \pounds_{u(\rigging + \Phi_{\star} (\gauge))} g 
\right )
& = \frac{1}{2} u \Phi^{\star} \left (  
\pounds_{\rigging} g \right )
+ \Phi^{\star} \left ( g(\rigging, \cdot) \otimes_s  du
\right )  
+ \frac{1}{2} \Phi^{\star} \left ( \pounds_{u\Phi_{\star} (\gauge))} g 
\right ) \\
& = u \Y 
+ \Phi^{\star} \left ( g(\rigging, \cdot) \right )
\otimes_s  du + \frac{1}{2} \pounds_{u \gauge} \gamma 
\\
& = u \Y 
+ \ellc \otimes_s du + \frac{1}{2} \pounds_{u \gauge} \gamma 
= \G_{(u,\zeta)} (\Y).
\end{align*}
\end{proof}

A gauge transformation obviously induces a  transformation
on the tensor fields $P,n, \nn$,  on the volume form 
and on the ``second fundamental form'' $\Kn$. The explicit transformation
laws for $P$, $\n$, $\nn$ and $\Kn$ were obtained in \cite{MarsGRG} (and 
for embedded data in \cite{MarsSenovilla1993}). We recall the result and 
add the gauge behaviour of the volume form.
\begin{lemma}
\label{TransOther}
Let $\metdata$ be metric hypersurface data. Under a gauge
transformation with gauge parameters $(u,\gauge)$, the tensors
$P$, $n$, $\nn$ and volume form $\etaell$ transform as
 \begin{align}
 \G_{(u,\gauge)} (P)  & =   P  + \nn \gauge \otimes \gauge - 2 \gauge \otimes_s n, 
\label{gaugeP} \\
  \G_{(u,\gauge)} (n)  & =   u^{-1} (  n - \nn \gauge ), \label{Gaugen} \\
  \G_{(u,\gauge)} (\nn)  & =  u^{-2} \nn,  \label{gaugenn}  \\
\G_{(u,\gauge)} (\Kn) & = u^{-1} \Kn,  \label{gaugeKn} \\
\G_{(u,\gauge)} (\etaell) & = |u| \, \etaell  \label{volformgauge}
\end{align}
where $|u|$ denotes absolute value of $u$.
\end{lemma}

\begin{proof}
Only (\ref{volformgauge}) remains to be shown.
We need
to compute the determinant 
\begin{align}
\mbox{det} & (\G_{(u,\gauge)} (\A) )
= \mbox{det} \left ( \begin{array}{cc}
\G_{(u,\gauge)} (\gamma) & \G_{(u,\gauge)} (\ellc)^T  \\
(\G_{(u,\gauge)} (\ellc)) & \G_{(u,\gauge)} (\ll ) 
\end{array}
\right ) \nonumber \\
& =
\mbox{det} \left ( \begin{array}{cc}
\gamma & u ( \ellc + \gamma(\gauge,\cdot) )^T \\
u ( \ellc + \gamma(\gauge,\cdot) ) 
& u^2 \left ( \ll + 2 \ellc(\gauge) + \gamma(\gauge,\gauge) \right )
\end{array}
\right )  \nonumber \\
& 
= \mbox{det} \left ( \begin{array}{cc}
\gamma & u \ellc^T  \\
u ( \ellc + \gamma(\gauge,\cdot) )^{T} 
& u^2 \left ( \ll + \ellc(\gauge) \right )
\end{array}
\right )
= \mbox{det} \left ( \begin{array}{cc}
\gamma & u \ellc^T  \\
u  \ellc  
& u^2  \ll  
\end{array}
\right ) = u^2 \det(\A) \label{changeA}
\end{align}
where in the third equality we added to the last column 
the first three columns applied to $u \gauge$ (hence without changing
the determinant)  and in the fourth one we applied a similar procedure, this
time acting on rows. (\ref{volformgauge}) then follows
directly from the definition (\ref{volformdef}).

\end{proof}

A natural question concerning gauge theories is to determine
which gauge parameters leave the data invariant, and provide an interpretation
for them. This is achieved in the following Lemma (see also the Remark
that follows).

\begin{lemma}
\label{invariance}
Let $\{\N,\gamma,\ellc,\ll\}$ 
be metric hypersurface data and
$(u,\gauge)$ a gauge parameter. $(u,\gauge)$ leaves invariant the data at $p \in \N$
if and only if the following happens
\begin{itemize}
\item[(i)] If $p$ is a null point, then  $u|_p =1$ and $\gauge|_p =0$.
\item[(ii)] If $p$ is non-null point then either $(u,\gauge)|_p = (1,0)$
or $(u,\gauge) |_p = (-1, - 2 \ell^{\sharp} |_p)$.
\end{itemize}
\end{lemma}

\begin{proof}
 The gauge transformation (\ref{gauge}) implies that
$\{\N,\gamma,\ellc,\ll\}$ stays invariant at $p$ under 
the gauge
transformation of $(u,\gauge)$ if and only if (we drop
the reference to $p$ for simplicity).
\begin{eqnarray}
\gamma (\gauge,\cdot) = \frac{1- u}{u} \ellc, \label{inv1}\\
2 \ellc(\gauge) + \gamma (\gauge,\gauge) = \frac{1 - u^2}{u^2} \ll.
\label{inv2}
\end{eqnarray}
If $p$ is a null point, then applying (\ref{inv1}) to $n$
gives $0 = (1-u) \ellc(n) = (1-u)$. Thus $u=1$ and (\ref{inv1})
implies $\gauge = \alpha n$ for some constant $\alpha$.
Equation (\ref{inv2}) and $\ellc(n)  = 1$  forces
$\alpha=0$ and item (i) is proved. 

If $p$ is a non-null point then (\ref{inv1}) yields
\begin{eqnarray}
\gauge  = \frac{1- u}{u} \ell^{\sharp},
\label{expV}
\end{eqnarray}
which inserted into (\ref{inv2}) gives
\begin{eqnarray*}
( 1 -u ) (1+ u) \left ( (\gamma^{\sharp}(\ellc,\ellc) - \ll \right ) = 0.
\end{eqnarray*}
From Lemma \ref{Pnnn_non-null} we know that 
$(\ll - \gamma^{\sharp}
(\ellc,\ellc) )\neq 0$, so either $u=1$ (and then (\ref{expV}) gives 
$\gauge=0$ ) 
or else
$u=-1$ (and then $\gauge =  - 2 \ell^{\sharp} := \gauge_0$). 
It only remains to show that in the latter case
$\ll$ is also gauge invariant:
\begin{align*}
\G_{(-1,\gauge_0)} (\ll) = \ll + 2 \ellc(\gauge_0) + \gamma(\gauge_0,\gauge_0)
= \ll - 4 \gamma^{\sharp} (\ellc,\ellc) + 4 \gamma^{\sharp} (\ellc,\ellc)
= \ll.
\end{align*}
This concludes the proof of  item (ii). 
\end{proof}

\begin{remark}
In the context of the matching problem between spacetimes 
the fact that  the hypersurface data remains invariant under two different
gauge transformation in the case of non-null data is related to the fact that any matching always admits a ``complementary matching'' \cite{FayosSenovilla}.
More specifically, consider two spacetimes $(M_1, g_1)$ and 
$(M_2,g_2)$ and two respective hypersurfaces $\Sigma_1$, $\Sigma_2$ 
that separate each spacetime in two domains $V_1^{\pm} \subset M_1$ and
$V_2^{\pm} \subset \M_2$.  Assume that neither $\Sigma_1$ nor
$\Sigma_2$ have any null points. If $V_1^{+}$ can be matched to 
$V_2^{-}$ then automatically the ``complementary'' matching
of $V_1^{-}$ with $V_2^+$ is also possible. This is because the matching of two
spacetimes only depends on the continuity of the hypersurface data, and the
two gauge transformations that preserve the data correspond to each one of the
two choices. When the data has null points, this ``complementary'' matching ceases to exists (except possibly in highly symmetric situations). This was first discussed in \cite[Proposition 1]{MarsSenovillaVera}. At the 
level of the hypersurface data, the reason behind this non-existence
of the complementary matching  is that the presence of null points
reduces the set of gauge transformations that leave the data invariant 
to the identity element, so one no-longer can change the orientation of the
rigging without altering the data.
\end{remark}

\begin{remark}
\label{interRem}
A geometric interpretation of this lemma can be given in the case
of embedded hypersurface data. Let $\metdata$ be  embedded
with rigging $\rigging$ and assume $p$ is a non-null point.
In order to apply Proposition \ref{changerig} we need
to determine $\Phi_{\star} (\ell^{\sharp})$. 
Lemma \ref{Pnnn_non-null} gives 
$\ell^{\sharp} = - ( \ll - \gamma^{\sharp} (\ellc,\ellc)) n$ so that
%
\begin{align*}
\Phi_{\star} (\ell^{\sharp}) & =
- \left ( \ll - \gamma^{\sharp} (\ellc,\ellc) \right ) \Phi_{\star} (n) 
=  - \left (\ll - \gamma^{\sharp} (\ellc, \ellc) \right ) 
\left ( \normal - \nn \rigging
\right ) \\
& = - \frac{1}{\nn} \normal + \rigging,
\end{align*}
where in the second equality we used (\ref{n})
in Proposition \ref{propembedded} and in the last one 
Lemma \ref{Pnnn_non-null} in order to replace
$\ll - \gamma^{\sharp}(\ellc,\ellc)$ in terms of $\nn$. 
Then,  Proposition \ref{changerig} provides
the gauge transformed rigging vector
\begin{align*}
\G_{(-1,\gauge_0)} (\rigging) = - \left ( \rigging -2 
\Phi_{\star} (\ell^{\sharp}) \right )
= \rigging - \frac{2}{\nn}   \normal.
\end{align*}
We finally use Remark \ref{nn} to write everything in terms
of ambient objects.
\begin{align*}
\G_{(-1,\gauge_0)} (\rigging) = 
\rigging - \frac{2}{g(\normal, \normal)} \normal.
\end{align*}
Thus, the gauge transformation turns out to be
a reflection with respect to the
hyperplane normal to $\normal$ (i.e. with respect to the
tangent plane $T_p \Phi(\N)$). So, we obtain a geometrically neat 
interpretation of the gauge transformation that leaves the metric
hypersurface data invariant, and we may understand why null points behave
very differently compared to  non-null points. The fact that at
null points the data remains invariant only under the identity element  
indicates, in particular, that the reflection along a non-null 
hyperplane does not admit any limit when the hyperplane approaches a null
one. It is of interest to see explicitly  what goes wrong along the limit.
We consider the following very simple situation. Let 
$o$ be the origin of the two dimensional Minkowski spacetime
$(\mathbb{M},g_{\mathbb{M}})$
and choose $\rigging = \partial_t$ in Minkowskian coordinates
$\{ t, x\}$. Consider the line $t = a x$ with $a \in \mathbb{R}$,
$a^2 \neq 1$. The normal vector satisfying $g_{\mathbb{M}} (\normal,
\rigging ) = 1$ is $\normal = - ( \partial_t + a \partial_x)$ and the
reflection of $\rigging$ is
\begin{align*}
\G_{(-1,\gauge_0)} (\rigging) = 
\rigging - \frac{2}{g(\normal, \normal)} \normal
= \partial_t + \frac{2}{a^2 -1} \left ( \partial_t + a \partial_x
\right ) = \frac{1}{a^2-1} \left ( (a^2+1) \partial_t + 2 a \partial_x \right ).
\end{align*}
Observe that the limit $a \rightarrow \pm 1$ is singular, as expected.
Even more, after removing the divergent factor $(1-a^2)^{-1}$ the
resulting vector is  $(a^2+1) \partial_t + 2a \partial_x$ which admits
a smooth limit as $a \rightarrow \pm 1$. However, at the limit, this vector
is $2( \partial_t \pm \partial_x)$ which is {\it tangent} to the corresponding
null hyperplane $t = \pm x$ and hence does not define a rigging.
In summary, we find (at least in this example) that 
the reason why the invariance gauge parameter 
$(-1,\gauge_0)$ does not survive at null points is that 
that rigging transformed  vector
$\G_{(-1,\gauge_0)} (\rigging)$  not only diverges but, in addition,
becomes tangent to the hyperplane  when the
hyperplane becomes null. 
\end{remark}

Lemma \ref{invariance} and Remark \ref{interRem} justify the following definition
\begin{definition}
The gauge parameter $(u,\gauge)$ is called {\bf orientation preserving} if $u>0$
everywhere.
\end{definition}

\begin{corollary}
The only orientation preserving gauge parameter $(u,\gauge)$ leaving invariant 
metric hypersurface data 
$\{\N,\gamma,\ellc,\ll\}$ is the trivial one $(u=1, \gauge=0)$.
\end{corollary}

\section{Metric hypersurface connection}
\label{Sect:Connection}
In \cite{MarsGRG} a torsion-free connection $\nablao$  
was introduced for embedded
hypersurface data. The expression of its Christoffel
symbols made sense for arbitrary metric hypersurface data (i.e.
they involved only $\gamma$, $\ellc$ and $\ll$, but not $\Y$).
This motivated the name {\bf metric hypersurface connection} 
for $\nablao$. However, in \cite{MarsGRG} 
it was not proved 
that such Christoffel symbols
indeed define a connection when the metric hypersurface data is not necessarily
embedded.
In this section, we define this connection at a fully abstract
level, following a strategy similar than for the standard proof
of existence and uniqueness of the Levi-Civita connection.
More specifically we prove that 
for any metric hypersurface data $\metdata$ 
there exists a unique connection $\nablao$ fulfilling certain identities 
involving $\nablao \gamma$ and $\nablao \ellc$
(see Proposition \ref{ExistConn}), which can be regarded as the replacement
to our setting of the metricity condition that characterizes the Levi-Civita connection.

First we 
introduce the following 2-covariant tensors on $\N$,
\begin{align}
\U  & \defi  \frac{1}{2} \pounds_{n} \gamma
 + \ellc \otimes_s d \nn  \label{defU} \\
F & \defi \frac{1}{2} d \ellc. \label{defF}
\end{align}
It is also useful to
define a one-form $\bm{\sone}$ by
\begin{align}
\bm{\sone} =  i_{n} F
\label{sone}
\end{align}
or, in coordinates, $\sone_b = F_{ab} n^a$. The following expressions
will be used often.
\begin{lemma}
\label{SomeIden}
Let $\hypdata$ be hypersurface data. Then
\begin{align}
\pounds_{n} \ellc & = 2 \bm{\sone} - d (\nn \ll) \label{poundsellc} \\
U (n, \cdot ) & = - \nn \bm{\sone} + \frac{1}{2} d\nn 
+ \frac{1}{2} (\nn)^2 d \ll \label{Un} \\
d \bm{\sone} & = \pounds_{n} F. \label{poundsF}
\end{align}
\end{lemma}

\begin{proof} 
Recall the Cartan identity $\pounds_{X} \omega = i_X d \omega 
+ d (i_X \omega)$ for any $p$-form. Applying this to $\ellc$
and using (\ref{EqP2bis}),
\begin{align*}
\pounds_{n} \ellc = i_n d \ellc + d ( \ellc(n))
= 2 i_n F - d (\nn \ll),
\end{align*}
which is (\ref{poundsellc}). Moreover,
\begin{align*}
\U(n, \cdot ) & = \frac{1}{2} \pounds_{n} ( \gamma(n,\cdot))
+ \frac{1}{2}  \ellc(n)  d \nn + \frac{1}{2} n(\nn) \ellc
= \frac{1}{2} \pounds_{n} (-\nn \ellc)
+ \frac{1}{2}  \ellc(n) d \nn + \frac{1}{2} n (\nn) \ellc \\
& = - \frac{1}{2} \nn \pounds_{n} \ellc 
+ \frac{1}{2} \ellc(n) d \nn 
 = - \nn \bm{\sone} + \frac{1}{2} \left ( \nn d (\nn\ll) +
(1- \nn\ll) d \nn \right ),
 \end{align*}
where in the second equality we used (\ref{EqP1bis})
and in the last equality (\ref{poundsellc}). This proves (\ref{Un}).
The third is a direct consequence of 
the Cartan identity applied to $F$ after using $dF =0$
and $\bm{\sone} = i_n F$.
\end{proof}

We also need to recall
the following result proved in 
\cite{MarsGRG} (Lemma 3).
\begin{lemma}
\label{Lemma3}
Given
a one-form field $\omega \in \X^{\star}(\N)$ and a scalar
$f \in {\mathcal F}(\N)$, there exists a vector field $V \in \X(\N)$ satisfying
$\gamma(V,\cdot)= \omega$ and $\ellc(V)=f$, if and only if 
\begin{align}
\omega(\n) + \nn f =0.
\label{compaCon}
\end{align}
Moreover, such $V$ is unique and given by
$V = P(\omega,\cdot) + f \n$. 
\end{lemma}

As stated, the {\bf metric hypersurface
connection} is introduced in the next Proposition.
\begin{proposition}
\label{ExistConn}
Let $\{ \N,\gamma,\ellc,\ll\}$ be metric hypersurface data. 
There exists a unique torsion-free connection $\nablao$, called
{\bf metric hypersurface connection} satisfying the equations
\begin{align}
(\nablao_X \gamma) (Z,W) & =
- \U(X,Z)\ellc(W) 
-  U(X,W)\ellc(Z), 
  \label{Cond1} \\
(\nablao_X \ellc) (Z) + (\nablao_Z \ellc) (X) & = - 2 \ll \U(X,Z),
\label{Cond2}
\end{align}
for all $X,Z,W \in \X(\N)$.
\end{proposition}

\begin{proof}

We prove the Proposition by providing an explicit expression for
$\nablao_X Z$. More specifically,
we provide expressions for
$\gamma(\nablao_X Z, \cdot)$ and $\ellc(\nablao_X Z)$ that
can be considered as a generalization of the Koszul formula for
the Levi-Civita connection (they reduce to it when
$\ellc= \ll=0$). Given $X,Z \in \X(\N)$ consider the one-form and scalar
\begin{align}
\omega_{X,Z} (W)  \defi 
&
U(X,Z) \ellc(W)  + \frac{1}{2} \left ( \frac{}{} \gamma([W,X],Z)
+ \gamma([W,Z],X) + \gamma(W,[X,Z] ) \right ) \nonumber \\
& + \frac{1}{2} \left ( \frac{}{} X (\gamma(W,Z)) +
Z (\gamma(W,X)) - W (\gamma(X,Z)) \right ), \label{defomega} \\
f_{X,Z}  \defi & X ( \ellc(Z)) - F(X,Z) + \ll \U(X,Z). \label{deff}
\end{align}
It is immediate that the first expression is indeed a one-form (i.e. 
satisfies $\omega_{X,Z} (h W) = h \omega_{X,Z}(W)$ for any scalar $h$).
By the previous lemma, these objects define a (unique) vector field $V_{X,Y}$ provided
\begin{align*}
\omega_{X,Z}(\n) + \nn f_{X,Z} =0,
\end{align*}
which we show next. Inserting (\ref{EqP1bis}) and
(\ref{EqP2bis})
\begin{align*}
\omega_{X,Z}(\n) 
 = & \, 
U(X,Z) \left ( 1 - \nn \ll \right ) 
+ \frac{1}{2} \left ( \frac{}{}
\gamma(\pounds_{\n} X,Z)
+ \gamma(\pounds_{\n} Z,X) - \nn \ellc([X,Z]) \right ) \\
& + \frac{1}{2} \left ( \frac{}{}  X ( - \nn \ellc (Z)  )  + 
Z  ( - \nn \ellc (X)  ) - \n (\gamma(X,Z)) \right ) \\
 = &  \, 
U(X,Z) \left ( 1 - \nn \ll \right ) 
+ \frac{1}{2} \left ( \frac{}{} - (\pounds_{\n} \gamma) (X,Z)
- \nn \ellc([X,Z]) - \nn X (\ellc(Z)) \right . \\
& \left . \frac{}{}- \nn Z (\ellc(X)) 
 - (d\nn \otimes \ellc) (X,Z) 
 -  (\ellc \otimes d \nn) (X,Z) \right ) \\
= & \,  - \nn \left ( \ll U(X,Z)
+ \frac{1}{2} \left ( \frac{}{} \ellc([X,Z]) + X(\ellc(Z))
+ Z(\ellc(X)) \right ) \right ) \\
= & \, - \nn \left ( \ll U(X,Z)  + X(\ellc(Z)) -  F(X,Z)
\right ) = - \nn f_{X,Z}
\end{align*}
where in the third equality we inserted $\pounds_{\n} \gamma$ in terms
of $\U$ using (\ref{defU}) 
and in the fourth one, we applied (to $\ellc$) the well-known general
identity  for one-forms
\begin{align}
\omega([X,Z]) = X(\omega(Z)) - Z(\omega(X)) - (d\omega) (X,Z), 
\label{idenw} 
\end{align}
and used (\ref{defF}).
We can therefore define the vector field 
 $V_{X,Z} := P(\omega_{X,Z},\cdot) + f_{X,Z} \n$
and note that
\begin{align*}
 \gamma(V_{X,Z}, \cdot) = \omega_{X,Z}, \quad
\quad
\ellc(V_{X,Z}) = f_{X,Z}
\end{align*}
hold as a consequence of Lemma \ref{Lemma3}.
It turns out that this 
operation defines a covariant derivative, which we denote by
$\nablao$, i.e.
\begin{align*}
\nablao_{X} Z := V_{X,Z}.
\end{align*}
Indeed $\omega_{h X,Z} = h \omega_{X,Z}$ and $f_{hX,Z} = h f_{X,Z}$ are 
easy 
consequences of their definition. Thus $V_{hX,Z} = h V_{X,Z}$ and the
property of linearity of the covariant derivative holds. Similarly
\begin{align*}
  \omega_{X,h W} = X(h) \gamma(W,\cdot) + h \omega_{X,W}, \quad \quad
f_{X,hW} = X(h) \ellc(W)+ h f_{X,W}
\end{align*}
are also direct consequences of the definitions. Since, by
(\ref{EqP4bis}),
\begin{align*}
P(\cdot,\gamma(W,\cdot)) +  \ellc(W) \n = W 
\end{align*}
we conclude that $V_{X,hW} = X(h) W + h V_{X,W}$, and the Leibniz rule
of the covariant derivative is proved. At this point, we have defined a
covariant derivative on $\N$. It remains to show that it is torsion-free,
satisfies (\ref{Cond1})-(\ref{Cond2}) and no other covariant derivative
has these properties. For the torsion-free condition we compute
\begin{align*}
\nablao_X Z - \nablao_Z X & = 
P(\omega_{X,Z} - \omega_{Z,X}, \cdot)  + \left ( f_{X,Z} - f_{Z,X} \right )
\n  \\
& =  P(\gamma([X,Z],\cdot),\cdot)
+ \left ( X(\ellc(Z)) - Z (\ellc(X)) - 2 F(X,Z) \right ) \n \\
& = 
[X,Z]   + \left ( - \ellc([X,Z]) + X(\ellc(Z)) - Z (\ellc(X)) - 2 F(X,Z) \right ) \n \\
& =
[X,Z]
\end{align*}
where in the third equality we used (\ref{EqP4bis})
and in the fourth one (\ref{idenw}). To show (\ref{Cond1})  we compute
\begin{align*}
(\nablao_X \gamma ) (Z,W) & = X (\gamma(Z,W)) - 
\gamma(\nablao_X Z, W) - \gamma(Z,\nablao_X W) \\
& =
X( \gamma(Z,W)) - \omega_{X,Z} (W) - \omega_{X,W} (Z) \\
& = 
- U(X,Z) \ellc(W) - U(X,W) \ellc(Z)
\end{align*}
after inserting the expression for $\omega_{X,Z}$. Concerning (\ref{Cond2})
\begin{align}
(\nablao_X \ellc) (Z) = X( \ellc(Z)) - \ellc(\nablao_X Z)
= X ( \ellc (Z)) -  f_{X,Z} = 
F(X,Z) - \ll U(X,Z)
\label{Cond3}
\end{align}
from which both (\ref{Cond2}) and the consistency condition
$(\nablao_X \ellc)(Z) - (\nablao_Z \ellc)(X) = 2 F(X,Z)=  d \ellc (X,Z)$ 
follows. Finally, we need to show uniqueness. Let $\nablao{}'$ be another torsion-free
connection satisfying (\ref{Cond1})-(\ref{Cond2}). Define the difference
tensor $S(X,Z) := \nablao_X Z - \nablao{}'_X Z$.  From
\begin{align*}
\gamma(\nablao_X Z,W) + \gamma(Z,\nablao_X W) = 
X( \gamma(Z,W) ) - (\nablao_X \gamma) (Z,W)
\end{align*}
and the similar equation for $\nablao{}'$ one finds, after subtraction,
\begin{align*}
\gamma(S(X,Z),W) + \gamma(Z,S(X,W)) =0.
\end{align*}
Writing three copies of this
\begin{align*}
\gamma(S(X,Z),W) + \gamma(Z,S(X,W)) & =0, \\
\gamma(S(Z,W),X) + \gamma(W,S(Z,X)) & =0, \\
\gamma(S(W,X),Z) + \gamma(X,S(W,Z)) & =0, 
\end{align*}
we can apply the standard argument of 
adding the first two and subtracting the third. The result is,
using  the
symmetry of $S(X,Z)$ and the fact that $W$ is arbitrary
\begin{align}
\gamma(S(X,Z),\cdot ) = 0.
\label{ort1}
\end{align}
On the other hand, from
\begin{align*}
\ellc(\nablao_X Z + \nablao_Z X) & =
X( \ellc (Z)) + Z ( \ellc(X)) - (\nablao_X \ellc)(Z) -
(\nablao_Z \ellc ) (X) \\
& = 
X( \ellc (Z)) + Z ( \ellc(X)) + 2 \ll U(X,Z)
\end{align*}
and the same expression for $\nablao{}'$, it follows, after subtraction 
\begin{align}
\ellc(S(X,Z) ) =0. \label{ort2}
\end{align}
By Lemma \ref{Lemma3}, the only vector $S(X,Z)$ satisfying (\ref{ort1})-(\ref{ort2})
is the zero vector, which concludes the uniqueness part.
\end{proof}

The connection coefficients of the connection $\nablao$ in a local coordinate
system $\{ x^a \}$ of $\N$, defined as usual by
$\nablao_{\partial_{b}} \partial_{a} \defi \Gamo{}^c{}_{ab} \partial_c$ read
\begin{align}
\Gamo{}^{c}{}_{ab} = 
\frac{1}{2} P^{cd} \left ( \partial_a \gamma_{bd} 
+ \partial_b \gamma_{ad} - \partial_d \gamma_{ab} \right )
+ \frac{1}{2}  n^c  \left (\partial_a \ell_b
+ \partial_b \ell_a \right ). \label{GambPropo}
\end{align}
Indeed, 
\begin{align*}
\gamma(\nablao_{\partial_b} \partial_a, \partial_d) & = 
\omega_{\partial_b,\partial_a}
(\partial_d)  =
U_{ba} \ell_d + \frac{1}{2} \left ( \partial_b \gamma_{ad}
+ \partial_{a} \gamma_{bd} - \partial_{d} \gamma_{ba} \right ), \\
 \ellc(\nablao_{\partial_b} \partial_a)  & = f_{\partial_b,\partial_a}
 = \partial_b \ell_a - 
\frac{1}{2} \left ( \partial_b \ell_a - \partial_a \ell_b \right ) 
+ \ll U_{ba} = 
\frac{1}{2} \left ( \partial_b \ell_a + \partial_a \ell_b \right )
+ \ll U_{ab},
\end{align*}
after using the definitions of $\omega_{XZ}$
and $f_{XZ}$  in (\ref{defomega})-(\ref{deff})
and the fact that $F_{ba} = \frac{1}{2} ( \partial_b \ell_a -
\partial_a \ell_b)$. Applying Lemma \ref{Lemma3}, expression
(\ref{GambPropo}) follows, after taking into account
(\ref{EqP2}),  i.e. that
$P^{cd} \ell_d + \ll n^c =0$.

As already mentioned, the explicit form of $\Gamo{}^c{}_{ab}$ appeared  in \cite{MarsGRG}, but no proof that this defines a connection at the fully abstract
level was given. Proposition \ref{ExistConn} accomplishes this.

An identity involving the 
contraction $\Gamo^{c}_{ca}$ was used in \cite{MarsGRG}.   
We write such identity in the form of a lemma for later use, and add
the proof for completeness.
\begin{lemma}
\label{traceGamo}
Let $\metdata$ be metric hypersurface data. In any coordinate system
$\{ x^a \}$ it holds
\begin{align*}
\Gamo{}^{c}_{ca} = - \frac{1}{2} \nn \partial_a \ll + \sone_a 
+ \frac{1}{2 (\mbox{det} (\A))} \partial_a \left ( \mbox{det} (\A)
\right ).
\end{align*}
\end{lemma}
\begin{proof}
Directly from (\ref{GambPropo})
\begin{align*}
\Gamo{}^c_{ca} &= 
\frac{1}{2} P^{cd} \partial_a \gamma_{cd}
+ \frac{1}{2} n^c \partial_a \ell_c + \frac{1}{2} n^c \partial_c \ell_a
= 
\frac{1}{2} P^{cd} \partial_a \gamma_{cd}
+ n^c \partial_a \ell_c + \frac{1}{2} n^c \left ( \partial_c \ell_a
- \partial_a \ell_c \right ) \\
& = \frac{1}{2} P^{cd} \partial_a \gamma_{cd}
+ n^c \partial_a \ell_c + n^c F_{ca}
\end{align*}
Now use
\begin{align*}
\frac{1}{2 \, \mbox{det}( \A)} \partial_a (\mbox{det} (\A)) = 
\frac{1}{2} (\A^{\sharp})^{\alpha\beta} \partial_a \A_{\alpha\beta}
= \frac{1}{2} P^{cd} \partial_a \gamma_{cd} 
+ n^c \partial_a \ell_c + \frac{1}{2} \nn \partial_a \ll 
\end{align*}
where the indices $\alpha, \beta$ take values in $0,1, \cdots, m$. Hence
\begin{align*}
\Gamo{}^c_{ca} &= \frac{1}{2 \, \mbox{det}( \A)} \partial_a (\mbox{det} (\A)) 
- \frac{1}{2} \nn \partial_a \ll + \sone_a
\end{align*}
as claimed.
\end{proof}

The gauge freedom of metric hypersurface data implies that $\N$
is endowed  not just with one connection, but with a family of connections,
each one attached to a representative of the gauge class. It is therefore
necessary to determine the relation between two such connections. The
computation is most easily done in coordinates. We first start with 
the following simple identities.
\begin{lemma}
\label{SomeIden2}
Let $\metdata$ be metric hypersurface data and $W \in \X(\N)$ be any vector field
and $f \in {\mathcal F}(\N)$.
Then, 
\begin{align}
\pounds_{f W} \gamma = f \pounds_{W} \gamma 
+ 2 \left ( W^{\flat} \otimes_s d f \right )
\label{Lief}
\end{align}
where $W^{\flat} := \gamma(W, \cdot)$.
Moreover, in any any local coordinate system of $\N$ it holds
\begin{align}
W^d \left ( \partial_a \gamma_{bd} 
+ \partial_{b} \gamma_{ad} - \partial_d \gamma_{ab} \right )
= -\pounds_{W} \gamma_{ab} + \partial_a W^{\flat}{}_b + \partial_b W^{\flat}{}_a.
\label{ConvLie}
\end{align}
\end{lemma}

\begin{proof}
Both follow directly by using the explicit expression for the Lie derivative
in coordinates and the Leibniz rule. 
\end{proof}

\begin{proposition}
\label{changeGam}
  Let $\metdata$ be metric hypersurface data. Under a gauge transformation
  with gauge parameters $(u,\gauge)$, the difference tensor between the connections
  $\nablao$ and $\G_{(u,\gauge)}(\nablao)$ is
  \begin{align*}
    \left ( \G_{(u,\gauge)}(\nablao) - \nablao \right )   = & \,  
    \frac{1}{2u} \gauge \otimes \left ( \pounds_{u n} \gamma
- \nn \pounds_{u \gauge} \gamma + 2 u \ellc \otimes_s d \nn \right )   + \frac{1}{2u } n \otimes \left ( \pounds_{u \gauge} \gamma
    + 2 \ellc \otimes_s du \right ).
      \end{align*}
\end{proposition}

\begin{proof}
  We use a coordinate system with associated basis $\partial_a$.
We use Lemma \ref{TransOther} to evaluate the components of
  $\left ( \G_{(u,\gauge)}(\nablao) - \nablao \right )$ as
  \begin{align*}
     \Big (  \G_{(u,\gauge)}& (\nablao) - \nablao \Big ){}^c{}_{ab}
     = (\G_{(u,\gauge)}(\Gamo)){}^{c}{}_{ab} - \Gamo{}^c{}_{ab}
    = \frac{1}{2} \left ( \G_{(u,\gauge)} (P)^{cd}  - P^{cd} \right )
    \left ( \partial_a \gamma_{bd}
    + \partial_{b} \gamma_{ad} - \partial_d \gamma_{ab} \right ) \\
    &
    + \frac{1}{2 u } ( n^c- \nn \gauge^c) \left (
\partial_a \G_{(u,\gauge)}(\ell)_b +
\partial_a \G_{(u,\gauge)}(\ell)_b \right )
    - \frac{1}{2} n^c \left (
\partial_a \ell_b +
\partial_b \ell_a \right )
\\
 = & \, \frac{1}{2} \gauge^c \left ( - \nn \pounds_{\gauge} \gamma_{ab}
+ \nn \partial_a \gauge^{\flat}_b + \nn \partial_b \gauge^{\flat}_a
+ \pounds_{n} \gamma_{ab} -
\partial_a n^{\flat}_b - \partial_b n^{\flat}_a \right ) \\
& + \frac{1}{2} n^c \left (
 \pounds_{\gauge} \gamma_{ab} -
\partial_a \gauge^{\flat}_b - \partial_b \gauge^{\flat}_a \right )
+ \frac{1}{2 u } ( n^c- \nn \gauge^c) \left (
\partial_a \G_{(u,\gauge)}(\ell)_b +
\partial_b\G_{(u,\gauge)}(\ell)_a \right ) 
  - \frac{1}{2} n^c \left (
\partial_a \ell_b +
\partial_b \ell_a \right ) \\
 = &\,  \frac{1}{2} \gauge^c \left (
- \nn \pounds_{\gauge} \gamma_{ab} + \pounds_{n} \gamma_{ab}
- \frac{\nn}{u} \left ( (\ell_b + \gauge^{\flat}_b) \partial_a u 
+ (\ell_a + \gauge^{\flat}_a) ) \partial_b u  \right )
+  \ell_a \partial_b  \nn  + \ell_b \partial_a \nn  ) \right ) \\
& + \frac{1}{2} n^c \left ( \pounds_{\gauge} \gamma_{ab}
+ \frac{1}{u} \left ( (\ell_a + \gauge^{\flat}_a) \partial_b u
+ (\ell_b + \gauge^{\flat}_b) \partial_a u \right ) \right ) \\
 = & \, \frac{1}{2} \gauge^c \left ( - \frac{\nn}{u} \pounds_{u \gauge} \gamma_{ab}
+ \frac{1}{u} \pounds_{u n} \gamma_{ab} 
+ \ell_a \partial_b \nn + \ell_b \partial_a \nn 
\right )  + \frac{1}{2} n^c \left ( \frac{1}{u} \pounds_{u \gauge} \gamma_{ab}
+ \frac{1}{u} \left ( \ell_a \partial_b u
+ \ell_b \partial_a u \right ) \right ) 
  \end{align*}
  where in the second equality we used  (\ref{ConvLie})
and in the third one $n^{\flat}_a = - \nn \ell_a$. The fourth
one follows from (\ref{Lief}) in Lemma \ref{SomeIden2}.
  \end{proof}

Consider a non-null point $p \in \N$ in a metric hypersurface
data
$\metdata$. This condition is open, so there is an  open neighbourhood 
$U_p \subset \N$ of $p$ where all points are non-null. On this set
$\gamma$ is a metric, so we can consider its corresponding Levi-Civita
connection $\nabgam$. From the previous Lemma and exploiting the gauge
structure, we can easily obtain the relationship between
the metric hypersurface connection and $\nabgam$.
\begin{proposition}
\label{DifCon}
Let $\metdata$ be metric hypersurface data, $p \in \N$ a non-null point and
$U_{p}$ an open neighbourhood of $p$ where $\gamma$ is a metric. On
$U_p$, the metric hypersurface connection $\nablao$ and
the Levi-Civita connection $\nabgam$ of $\gamma$ are related by
\begin{align*}
\nablao = \nabgam  - \frac{1}{2(\ll - \gamma(\ell^{\sharp},\ell^{\sharp}))}
\ell^{\sharp} \otimes \pounds_{\ell^{\sharp}} \gamma
\end{align*}
where $\ell^{\sharp} := \gamma^{\sharp}(\ellc,\cdot)$.
\end{proposition}

\begin{proof}
Consider the gauge parameters $(u=1, \gauge = - \ell^{\sharp})$. The gauge transformed
data are
\begin{align*}
\G_{(1,-\ell^{\sharp})}(\ellc) =0, \quad \quad \G_{(1,-\ell^{\sharp})} (\ll)
= \ll - 2 \ellc( \ell^{\sharp}) + \gamma(\ell^{\sharp},\ell^{\sharp}) =
\ll - \gamma(\ell^{\sharp}, \ell^{\sharp} ).
\end{align*}
Thus, the connection $\G_{(1,-\ell^{\sharp})}(\nablao)$ is metric
(because of its defining property \ref{Cond1}) and hence agrees with
$\nabgam$. So,  $\nabgam$ is the metric hypersurface connection of
the data
\begin{align*}
(\N, \gamma, \ellc^{\prime}=0, 
\ll{}^{\prime} = \ll - \gamma(\ell^{\sharp}, \ell^{\sharp}) \neq 0 )
\end{align*}
for which, in addition $n'=0$, $\nn{}^{\prime} = \frac{1}{\ll{}^{\prime}}$.
Using the group structure and the fact that 
$(1,-\ell^{\sharp})^{-1} = (1, \ell^{\sharp})$ we have
\begin{align*}
\nablao = \G_{(1,-\ell^{\sharp})^{-1}}  (
\G_{(1,-\ell^{\sharp})} \nablao ) =
\G_{(1,\ell^{\sharp})}  ( \nabgam) =
\nabgam - \frac{\nn{}'}{2} \ell^{\sharp} \otimes \pounds_{\ell^{\sharp}} \gamma 
\end{align*}
where in the last equality we used Proposition (\ref{changeGam}) with 
$(u,\gauge)= (1, \ell^{\sharp})$, 
$\nablao = \nabgam$, $\ellc = 0$ and  $n =0$. Inserting $\nn{}' = \frac{1}{\ll{}'}
= \frac{1}{\ll - \gamma(\ell^{\sharp},\ell^{\sharp})}$, the result follows.
\end{proof}

Returning to the general case, the connection $\nablao$ is defined by its action on $\gamma$ and 
$\ellc$. It is necessary to know also how it acts on 
the associated tensors $P$ and $n$. The expressions were derived in
\cite{MarsGRG}, but using an indirect argument that exploited yet
another connection $\overline{\nabla}$ on codimension
one submanifolds called ``rigging connection'' \cite{MarsSenovilla1993, MarsGRG}. We include a self-contained 
derivation using only $\nablao$ for completeness.
\begin{lemma}
\label{identitiesnablao}
Let $\{\N,\gamma,\ellc,\ll\}$ be metric hypersurface data, 
then
the following identities hold.
\begin{align}
\nablao_{a} \gamma_{bc} & = - \ell_b \U_{ac} - \ell_c \U_{ab}, \label{nablaogamma} \\
\nablao_a \ell_b & = F_{ab} - \ll \U_{ab}, \label{nablaoll}\\
\nablao_a n^b & = - \nn n^b (d \ll)_a - \left( \nn P^{bf} + n^b n^f \right ) F_{af} + P^{bf} \U_{af},
\label{nablaon}  \\
\nablao_{a} P^{bc} & = - \left ( n^b P^{cf} + n^c P^{bf} \right ) F_{af} - n^b n^c \partial_a \ll, \label{nablaoP} 
\end{align}
\end{lemma}

\begin{proof}
The first one is just the expression of 
(\ref{Cond1})  in abstract index notation. The second has already been
obtained in \ref{Cond3}. In order to prove (\ref{nablaon}) we compute
\begin{align*}
0  = \nablao_{a} ( n^b \ell_b  + \nn \ll) &=
\ell_b \nablao_a n^b  + n^b \left ( F_{ab} - \ll \U_{ab}
\right ) + \nablao_a (\nn \ll) :=
\ell_b \nablao_a n^b  + q_a \\
0  = \nablao_a ( n^b \gamma_{bc} + \nn \ell_c ) 
 & =
\gamma_{bc} \nablao_a n^b + n^b \left ( - \ell_b \U_{ac} - \ell_c \U_{ab}
\right ) + \ell_c \nablao_a \nn + \nn F_{ac} - \nn \ll \U_{ac} \\
& = \gamma_{bc} \nablao_a n^b - \U_{ac} + \ell_c \left ( \nablao_a \nn
- n^b \U_{ab}  \right ) + \nn F_{ac} :=
\gamma_{bc} \nablao_a n^b + Q_{ac}.
\end{align*}
According to Lemma \ref{Lemma3}, the expression for $\nablao_a n^b$ must
be (there is no need to check explicitly that the compatibility condition
(\ref{compaCon}) holds true, since $\nablao_a n^b$ obviously  exists)
\begin{align*}
\nablao_a n^b &=  -P^{bf} Q_{af} - q_a n^b \\
& =  P^{bf} \left ( U_{af} - \ell_f \left ( \nablao_a \nn -  n^d \U_{da} 
  \right ) - \nn F_{af} \right ) 
+ n^b \left ( - n^f F_{af} + \ll n^f \U_{af} - \nablao_a (\nn \ll) \right ) \\
& = P^{bf} \left ( \U_{af} - \nn  F_{af} \right )
+ n^{b} \left (  \ll  \nablao_a \nn 
- n^f F_{af} - \ll \nablao_a \nn- \nn \nablao_a \ll \right )
\end{align*}
which is (\ref{nablaon}).
Similarly, 
\begin{align*}
0 = \nablao_a \left ( P^{bc} \gamma_{fc} + n^b \ell_f \right ) & =
\gamma_{fc} \nablao_{a} P^{bc}  + P^{bc} \nablao_a \gamma_{fc}
+ \nablao_a (n^b \ell_f) := \gamma_{fc}  \nablao_a P^{bc} 
+ Q^{b}{}_{af} \\
0 = \nablao_a \left ( P^{bc} \ell_c + \ll n^b \right ) & =
\ell_c \nablao_a P^{bc}  + P^{bc} \nablao_a \ell_c 
+ \nablao_a (\ll n^b) := \ell_c \nablao_a P^{bc} 
+ q^b{}_a.
\end{align*}
Applying Lemma (\ref{Lemma3}) it must be
\begin{align}
\nablao_a P^{bc} =  -P^{cf} Q^{b}{}_{af}  - q^b{}_a n^c.
\label{fromlemma}
\end{align}
Now,
\begin{align*}
P^{cf} Q^{b}{}_{af} & =
P^{cf} \left ( - P^{bd} \ell_f \U_{ad} - P^{bd} \ell_d \U_{af} 
+ \ell_f \nablao_a  n^b + n^b \nablao_a \ell_f \right ) \\
& = \ll n^c \left ( P^{bd} U_{ad} - \nablao_a n^b \right )
+ P^{cf} n^b \left ( \nablao_a \ell_f  + \ll U_{af} \right ) \\
& = \ll n^c \left ( P^{bd} U_{ad} - \nablao_a n^b \right )
+ P^{cf} n^b F_{af}
\end{align*}
and
\begin{align*}
q^b{}_a n^c
= n^c \left ( P^{bd} F_{ad} + \ll \left ( - P^{bd} \U_{ad} + \nablao_a n^b \right )+ n^b \nablao_a \ll 
\right ).
\end{align*}
Inserting into (\ref{fromlemma}) the terms $(-P^{bd} \U_{ad} + \nablao_a n^d)$
cancel out (this is why the explicit
expression of $\nablao_a n^b$ was not inserted), leaving
\begin{align*}
\nablao_a P^{bc} = - P^{cf} n^b F_{af}  - P^{bd} n^c F_{ad}
- n^b n^c \nablao_a \ll,
\end{align*}
which is (\ref{nablaoP}).
\end{proof}

\begin{corollary}
With the same assumptions
\begin{align}
(\nablao_n n)^b & =  P^{bf} \left ( - 2 \nn \sone_f + \frac{1}{2} (d \nn)_f
+ \frac{1}{2} (\nn)^2 (d \ll)_f  \right ) - \nn n(\ll) n^b, \label{ndern} \\
\ell_b \nablao_a n^b & = (1 -  \nn \ll) \left (
\sone_a  - \frac{1}{2} \nn  (d \ll)_a \right ) - \frac{1}{2}
d ( \ll \nn)_a   \label{ldern}
\end{align}
\end{corollary}

\begin{proof}
For the first, contract (\ref{nablaon}) with $n^a$ and use (\ref{Un}), as
well as the definition of 
$\sone$. For the second contract (\ref{nablaon}) with $\ell_b$
and use $\ellc(n) = 1 - \nn \ll$ and $P(\ellc,\cdot) = - \ll n$ to obtain
\begin{align*}
\ell_b \nablao_a n^b & =  ( 1- \nn \ll) ( \sone_a - \nn (d \ll)_a )
+ \ll n^f ( - U_{af} + \nn F_{af} ).
\end{align*}
Inserting (\ref{Un}) and simplifying yields (\ref{ldern}).
\end{proof}

\begin{remark}
When $\metdata$ is null (so that $\nn =0$ everywhere), the vector field $n$ is automatically
autoparallel and affinely parametrized with respect to the
connection $\nablao$.
\end{remark}

It remains to compute the covariant derivative of the volume form.
\begin{lemma}
\label{dervolume}
\label{deretaell}
Let $\{\N,\gamma,\ellc,\ll\}$ be metric hypersurface data. Then  
\be
\nablao_X \etaell = \bm{\omega}(X) \etaell, \quad \quad
\mbox{where} \quad \quad
\bm{\omega}  \defi \frac{1}{2} \nn d \ll - \bm{\sone}.
\en
\end{lemma}

\begin{proof}
We work in coordinates. The components of the volume form are
(we assume without loss of generality a positively oriented chart)
\begin{align*}
\etaell_{b_1 \cdots b_m} = \sqrt{|\mbox{det} (A)|}  \epsilon_{b_1
\cdots b_m}. 
\end{align*}
Its covariant derivative in terms of the connection symbols
$\Gamo{}^{a}_{bc}$ of $\nablao$ has the standard form
\be
\nablao_a \etaell_{b_1 \cdots b_m} = \partial_{a}
\etaell_{b_1 \cdots b_m} - \Gamo{}^{c}_{ac} \etaell_{b_1 \cdots b_m}
= \left ( 
\frac{1}{2 \, \mbox{det}(\A)} \partial_a (\mbox{det} (\A))
- \Gamo{}^c_{ca} \right ) \etaell_{b_1 \cdots b_m}.
\en
Using Lemma \ref{traceGamo} yields
\begin{align*}
\nablao_a \etaell_{b_1 \cdots b_m} = 
\left (\frac{1}{2} \nn \partial_a \ll - \sone_a\right )
\etaell_{b_1 \cdots b_m}.
\end{align*}
\end{proof}

\section{Metric hypersurface curvature}
\label{Sect:Curvature}

In this section we explore a number of general properties of the
curvature of the metric hypersurface connection
$\nablao$.  
Our convention for the curvature operator $R^{\nabla}$
of any (torsion-free) connection $\nabla$ is
\be
R^{\nabla}(X,Y) Z \defi
\left ( \nabla_X \nabla_Y - \nabla_Y \nabla_X - \nabla_{[X,Y]} \right ) Z,
\en
and the curvature tensor is defined by
$R^{\nabla}(\bm{\alpha},Z,X,Y)$ (where $\bm{\alpha}$ is any one-form)
\be
\mbox{Riem}^{\nabla} ( \bm{\alpha},Z,X,Y) ) \defi
\bm{\alpha} \left ( R^{\nabla}(X,Y) Z \right ).
\en
The Ricci tensor  $\Ricc^{\nabla}$ 
is defined by contracting the first and third indices.
The curvature tensor of the metric hypersurface connection
$\nablao$ is denoted by  $\Riemo$ and
its Ricci tensor by 
$\Ricco$. When using abstract
index notation we write simply
$\Riemoin{}^{a}{}_{bcd}$ and
$\Riemoin_{ab}$. The trace of the curvature tensor
$\Riemo$ tensor with respect to the
first two indices (which does not vanish in general for torsion-free but
non-metric connections) is denoted by $\Vo$.
In the next proposition we find its explicit expression.

\begin{proposition}
\label{antisymRiemoProp}
Let $\metdata$ be metric hypersurface data. The traces of curvature tensor
of the metric hypersurface connection $\nablao$ satisfy
the following identities.
\begin{equation}
  \Ricco(X,Y) - \Ricco(Y,X) = \Vo(X,Y) =  \left ( d \bm{\sone} -
 \frac{1}{2} d \nn \wedge d \ll \right )(X,Y),
\quad  X,Y \in \X(\N).
\label{antisymRiemo}
\end{equation}
\end{proposition}

\begin{proof}
The well-known relationship $\Vo(X,Y)= \Ricco(X,Y) - \Ricco(Y,X)$
follows directly from the first Bianchi identity, so we only need to
show the second equality. 
In abstract index notation, the Ricci identity applied to 
the volume form $\etaell$ gives
\begin{align*}
\left ( \nablao_a \nablao_b - \nablao_b \nablao_a \right ) \etaell_{c_1
\cdots c_m} & = - \sum_{i=1}^m \Riemoin{}^{d_i}_{\,\,\,c_i a b}
\etaell_{c_1  \cdots d_i \cdots c_m} = 
- \Riemoin{}^{d}_{\,\,\,dab} \etaell_{c_1 \cdots c_m} \\
& = - \Vo_{ab} \etaell_{c_1 \cdots c_m}
= \Vo_{ba} \etaell_{c_1 \cdots c_m}
\end{align*}
where in the second equality we used the antisymmetry of $\etaell$ 
and in the last one the fact that $\Vo$ is antisymmetric.
We apply lemma \ref{deretaell}
\be
\nablao_b \etaell_{c_1 \cdots c_m} = \omega_b \etaell_{c_1 \cdots c_m}
\quad \quad  \Longrightarrow \quad \quad
\nabla_a \nablao_b \etaell_{c_1 \cdots c_m} = \left ( \nablao_a \omega_b 
+ \omega_a \omega_b \right ) \etaell_{c_1 \cdots c_m}
\en
and we conclude
\be
\Vo_{ab} = \nablao_b \omega_a - \nablao_a \omega_b =
\frac{1}{2} \left ( \nablao_b \nn \nablao_a \ll - \nablao_a \nn \nablao_b
\ll \right )
- \nablao_b \sone_a + \nablao_a \sone_b 
\en
after using the explicit form of $\omega_a$.
Since in abstract index notation $(d \bm{\omega})_{ab} = \nablao_a \omega_b
- \nablao_b \omega_a$ the result can be  written as
 (\ref{antisymRiemo}).
\end{proof}

In  the next Proposition we derive a several identities satisfied
by the curvature tensor of $\nablao$. 
\begin{proposition}
Let $\{\N,\gamma,\ellc,\ll\}$ be 
metric hypersurface data. Then, the 
curvature tensor $\Riemo$ satisfies the following identities:
\begin{align}
\ell_d \Riemoin{}^d{}_{acb}  =  & \nablao_a F_{cb} 
+ \nablao_c \left ( \ll \U_{ba} \right )-
\nablao_b \left ( \ll \U_{ca} \right ) \label{first} \\
\ell_d \Riemoin{}^d{}_{acb} n^c  = &  
\nablao_a \sone_b + \pounds_{n}  (\ll \U)_{ba} 
+  ( F_{bc} - \ll \U_{bc} ) P^{cf} (\U_{af} - \nn F_{af} ) \nonumber\\
& + \nn \ll \nablao_b \sone_a - \frac{1}{2} \nablao_b \nablao_a (\nn \ll) 
+  \frac{1}{2} \nablao_b (\nn \ll)  \sone_a \nonumber \\
&+ (\nn \ll -1 ) \left ( \left (  \sone_b - \frac{1}{2} \nn \nablao_b \ll
\right ) \left ( \sone_a - \nn \nablao_a \ll \right )
- \frac{1}{2} \nablao_b \left ( \nn \nablao_a \ll  \right ) \right )  
\label{second} \\
\Riemoin{}^d{}_{acb} n^a = &
n^d \left ( 2 \nablao_{[c} \sone_{b]}
+ \nablao_{[c} \ll \nablao_{b]}  \nn
+ 2 P^{af} F_{a[c} \U_{b]f} \right ) + P^{df} \left [ 
2 \nablao_{[c} \U_{b]f} + \nn \nablao_f F_{cb}  
\right . \nonumber \\
& \left . 
+  2 U_{f[c} \left ( \sone_{b]} 
- \nn \nablao_{b]} \ll \right ) 
+ F_{f[c} \left ( 2 \nn \sone_{b]} - \nablao_{b]} \nn -  
\nn{}^2 \nablao_{b]} \ll \right ) 
\right ] \label{third} \\
\ell_d \Riemoin{}^d{}_{acb} n^a = &  \left ( 1- \ll \nn \right )
\left ( 2 \nablao_{[c}  \sone_{b]}
+ 2 P^{af} F_{a[c} U_{b]f}   + \nablao_{[c} \ll \nablao_{b]} \nn
\right ) \nonumber \\
& + \ll \sone_{[c} 
\left ( \nablao_{b]} \nn + \nn{}^2 \nablao_{b]} \ll \right )
\label{fourth} \\
\ell_d \Riemoin{}^d{}_{acb} n^a n^c = &  \left ( 1- \ll \nn \right )
\left ( \pounds_{n} \sone_b - \sone_a P^{af} 
\left (\U_{bf} + \nn F_{bf} \right )  
+ \frac{1}{2} P^{af} F_{ba} \left ( \nablao_f \nn + \nn{}^2 \nablao_f \ll
\right ) \right ) \nonumber \\
& + \frac{1}{2} n(\ll) \left [ 
(1 - \ll \nn ) 
\nablao_b \nn - \ll \nn{}^2 \sone_{b} \right ]
- \frac{1}{2} n(\nn) \left [
 (1-\nn\ll)  \nablao_{b} \ll  + \ll \sone_b \right ].
\label{fifth}
\end{align}
\end{proposition}
\begin{proof}
Since $\nablao$ is torsion-free, the Ricci identity applied to $\ellc$
gives
\begin{align*}
\ell_d \Riemoin{}^d{}_{acb}
 & = \nablao_b \nablao_c \ell_a - \nablao_{c} 
\nablao_b \ell_a =
\nablao_b \left ( F_{ca} - \ll \U_{ca} \right ) -
\nablao_c \left ( F_{ba} - \ll \U_{ba} \right ) \\
& =
\nablao_a F_{cb} + \nablao_c \left ( \ll \U_{ba} \right )-
\nablao_b \left ( \ll \U_{ca} \right )
\end{align*}
where in the third equality we used the fact that $dF=0$, i.e.
$\nablao_{[a} F_{bc]} =0$. This shows (\ref{first}). 

For the second, we exploit the 
following general identity (easy to show)
valid  for any covariant two-tensor $V_{ab}$,
vector $W^c$ and any torsion-free covariant derivative 
(in particular $\nablao$).
\begin{align*}
W^{c} \left ( \nablao_{c} V_{ba} - \nablao_b V_{ca} \right )
= (\pounds_{W} V)_{ba} - \nablao_b \left ( V_{ca} W^c \right ) 
- V_{bc} \nablao_a W^c.
\end{align*}
Applying this to $V_{ba} = \ll \U_{ba}$ and to $W^c= n^c$ it follows
\begin{align*}
\ell_d \Riemoin{}^d{}_{acb} n^c & = 
\nablao_a ( n^c F_{cb} ) - F_{cb} \nablao_a n^c 
+ \pounds_{n}  (\ll \U)_{ba} - \nablao_b ( \ll \U_{ca} n^c )
- \ll U_{bc} \nablao_a n^c \\
& = \nablao_a \sone_b 
+ ( F_{bc} - \ll \U_{bc} ) \nablao_a n^c
+ \pounds_{n}  (\ll \U)_{ba} 
- \nablao_b ( \ll \U_{ca} n^c ).
\end{align*}
Inserting now (\ref{nablaon}) it follows
\begin{align}
\ell_d \Riemoin{}^d{}_{acb} n^c = &
\nablao_a \sone_b + \pounds_{n}  (\ll \U)_{ba} 
+ ( F_{bc} - \ll \U_{bc} ) 
( \sone_a - \nn \nablao_a \ll ) n^c \nonumber \\
& + ( F_{bc} - \ll \U_{bc} ) P^{cf} (\U_{af} - \nn F_{af} )
- \nablao_b ( \ll \U_{ca} n^c ) \nonumber \\
= & 
\nablao_a \sone_b + \pounds_{n}  (\ll \U)_{ba} 
+  ( F_{bc} - \ll \U_{bc} ) P^{cf} (\U_{af} - \nn F_{af} ) \nonumber\\
& 
- (\sone_b + \ll  \U_{bc} n^c) ( \sone_a - \nn \nablao_a \ll) 
- \nablao_b ( \ll \U_{ac} n^c ) \label{intermediate}
\end{align}
The elaborate the terms in the last line, both of which involve
$\ll \U_{ac}  n^c$. Recalling (\ref{Un}), we have
\begin{align*}
\ll U_{ac} n^c & = - \nn \ll \sone_a + \frac{1}{2} \ll \nablao_a \nn 
+ \frac{1}{2} (\nn)^2 \ll \nablao_a \ll \\
& = - \nn \ll \sone_a + \frac{1}{2} \nablao_a ( \nn \ll )
+ \frac{1}{2} \nn ( \nn \ll -1  ) \nablao_a \ll
\end{align*}
Taking $\nablao_b$ of this expression, a straightforward computation
yields
%
%
\begin{align*}
-  (\sone_b & + \ll  \U_{bc} n^c)( \sone_a - \nn \nablao_a \ll) 
- \nablao_b ( \ll \U_{ac} n^c )
=  \\
= \, &
\nn \ll \nablao_b \sone_a - \frac{1}{2} \nablao_b \nablao_a (\nn \ll) 
+  \frac{1}{2} \nablao_b (\nn \ll)  \sone_a \\
&+ (\nn \ll -1 ) \left ( \left (  \sone_b - \frac{1}{2} \nn \nablao_b \ll
\right ) \left ( \sone_a - \nn \nablao_a \ll \right )
- \frac{1}{2} \nablao_b \left ( \nn \nablao_a \ll  \right ) \right ). \\
\end{align*}
Inserting into (\ref{intermediate}) one obtains (\ref{second}).

We next prove  (\ref{third}). We use the
Ricci identity $\Riemoin{}^d{}_{acb} n^a =
2 \nablao_{[c} \nabla_{b]} n^d$, where brackets denote antisymmetrization.
The computation gets simplified if we define $Q_{b} := \sone_{b} - \nn \nablao_b
\ll$ and $T_{bf} := \U_{bf} - \nn F_{bf}$ so that 
(\ref{nablaon})  is written as
\begin{align*}
\nablao_b n^d = n^d Q_{b} + P^{df} T_{bf}
\end{align*}
and hence
\begin{align*}
\nablao_{[c} \nabla_{b]} n^d  = &
\nablao_{[c} \left ( n^d Q_{b]}  + P^{df} T_{b]f} \right ) \\
 = &  \left ( n^d Q_{[c} + P^{d f} T_{[c |f} \right ) Q_{b]}
+ n^d \nablao_{[c} Q_{b]} 
+ (\nablao_{[c} P^{df}) T_{b]f} + P^{df} \nablao_{[c} T_{b] f}.
\end{align*}
 Next we insert $\nablao_c P^{bf} =  n^b L^f{}_c + n^f L^{b}{}_c$
with $L^{f}{}_c :=  P^{fa} F_{ac} - \frac{1}{2} n^f \nablao_c \ll$ (see
\ref{nablaoP}) and 
obtain
\begin{align}
\nablao_{[c} \nabla_{b]} n^d  = & n^d \nablao_{[c} Q_{b]} 
+ P^{df} \left ( \nablao_{[c} T_{b] f} + T_{[c|f} Q_{b]} \right )
+ n^d L^{f}{}_{[c} T_{b] f}  
+ n^f L^{d}{}_{[c} T_{b]f} \nonumber \\
= & n^d \left ( \nablao_{[c} Q_{b]} +  L^{f}{}_{[c} T_{b] f}  \right )
+ P^{df} \left ( \nablao_{[c} T_{b] f} + T_{[c|f} Q_{b]} \right )
+ \frac{1}{2} L^{d}{}_{[c}  \left ( \nablao_{b]} \nn 
+ \nn{}^2 \nablao_{b]} \ll \right ) \nonumber \\
= & n^d \underbrace{\left ( 
\nablao_{[c} Q_{b]} +  L^{f}{}_{[c} T_{b] f} - \frac{1}{4} \nablao_{[c} \ll
\nablao_{b]} \nn \right )}_{{\mathcal B}_{cb}} \nonumber \\
& + P^{df} \underbrace{\left ( \nablao_{[c} T_{b] f} + T_{[c|f} Q_{b]} 
+ \frac{1}{2} F_{f[c} \left ( \nablao_{b]} \nn + \nn{}^2 \nablao_{b]} \ll
\right )
\right )}_{{\mathcal C}_{fcb}} \label{ExpBC}
\end{align}
where in the second equality we used
\begin{align}
n^f T_{bf} = n^f \left (\U_{bf} - \nn F_{bf} \right )
= \frac{1}{2} \left ( \nablao_{b} \nn + \nn{}^2
\nablao_b \ll \right )
\label{nT}
\end{align}
which follows directly from (\ref{Un}) and (\ref{sone}). We next compute
${\mathcal B}_{cb}$ and  ${\mathcal C}_{fcb}$. From the definitions   of $Q_{b}$,
$L^{f}{}_{c}$ and $T_{bf}$:
\begin{align}
& \nablao_{[c} Q_{b]} =  \nablao_{[c}  \sone_{b]} 
- \nablao_{[c} \nn \nabla_{b]} \ll \\
& L^{f}{}_{[c} T_{b]f} =   (P^{fa} F_{a[c} - \frac{1}{2} n^f \nablao_{[c} \ll 
)( U_{b]f} - \nn F_{b]f} ) 
= P^{fa} F_{a[c} U_{b]f} 
- \frac{1}{4} \nabla_{[c} \ll \nablao_{b]} \nn \\
& \Longrightarrow \quad \quad {\mathcal B}_{cb} =
\nablao_{[c}  \sone_{b]}
+ P^{af} F_{a[c} U_{b]f}   + \frac{1}{2} \nablao_{[c} \ll \nablao_{b]} \nn
\end{align}
where in the second line   (\ref{nT}) was used again. Concerning
${\mathcal C}_{fcb}$,
\begin{align}
  & \nablao_{[c} T_{b]f} = \nablao_{[c} \U_{b]f} - \nablao_{[c} \nn
  F_{b] f} - \nn \nabla_{[c} F_{b]f} = \nablao_{[c} \U_{b]f} -
  \nablao_{[c} \nn F_{b] f} + \frac{1}{2} \nn \nablao_f F_{cb}
  \nonumber \\
  & T_{[c|f} Q_{b]} = \U_{[f[c} (\sone_{b]} -\nn \nablao_{b]} \ll ) +
  \nn F_{f[c} (\sone_{b]} -\nn \nablao_{b]} \ll ) \quad \quad \quad
  \quad \quad \quad \quad \quad \quad \quad \quad
  \Longrightarrow \nonumber \\
  & 
  {\mathcal C}_{fcb} = \nablao_{[c} \U_{b]f} + \frac{1}{2} \nn
  \nablao_f F_{cb} + \U_{f[c} \left (\sone_{b]} -\nn \nablao_{b]} \ll
  \right )
+ F_{f[c} \left ( \nn \sone_{b]} 
- \frac{1}{2} \nablao_{b]} \nn - \frac{1}{2} \nn{}^2 \nablao_{b]} \ll
\right )
\label{mathcalC}
\end{align}
where in the second equality of the first line we used $d F=0$.
With these expressions for ${\mathcal B}_{cb}$ and ${\mathcal C}_{fcb}$, 
(\ref{ExpBC}) combined with  the Ricci identity for $\n$ yields (\ref{third}).

We turn our attention to (\ref{fourth}). This can be computed in two different
ways, namely contracting (\ref{first}) with $n^a$ or (\ref{third}) with
$\ell_d$. It is a non-trivial check that both methods
yield the same result. Here we only present the second method, which is
simpler. From $\ell_d n^d
= 1 - \nn \ll$ and $\ell_d P^{df} = -\ll n^f$ we have
\begin{align}
\frac{1}{2} \ell_d \Riemoin{}^d{}_{acb} n^a  = (1- \ll \nn ) {\mathcal B}_{cb}
- \ll n^f {\mathcal C}_{fbc}
\label{inter}
\end{align}
so we only need to compute the last contraction. From 
(\ref{mathcalC})
\begin{align}
n^f {\mathcal C}_{fcb} = & 
\nablao_{[c} \left ( U_{b]f} n^f \right )
+ U_{f[b} \left ( - \nablao_{c]} n^f - n^f \sone_{c]} + \nn n^f \nabla_{c]} \ll \right )
+ \frac{1}{2} \nn n^f \nablao_{f} F_{cb} \nonumber \\ 
& + \sone_{[c} \left (
\nn \sone_{b]} - \frac{1}{2} \nablao_{b]} \nn - \frac{1}{2} \nn{}^2 \nablao_{b]} \ll
\right ).
\label{nB}
\end{align}
Using respectively (\ref{Un}), (\ref{nablaon}) and (\ref{poundsF}),
the first three terms can be rewritten as
\begin{align*}
(a) \hspace{1cm} & \nablao_{[c} \left ( U_{b]f} n^f \right ) 
= - \nn \nablao_{[c} \sone_{b]} + \sone_{[c} \nablao_{b]} \nn 
+ \nn \nablao_{[c} \nn \nablao_{b]} \ll, \\
(b) \hspace{1cm} & \left ( - \nablao_{c} n^f - n^f \sone_{c} + \nn n^f \nablao_{c} \ll \right )  = 
- 2 n^f \sone_{c} 
+ 2 \nn n^f \nablao_{c} \ll -  P^{fa} \left ( \U_{ac} + \nn F_{ac} \right ), \\
(c) \hspace{1cm} & \frac{1}{2} n^f \nablao_f F_{cb}  = \frac{1}{2}
\pounds_{n} F_{cb} +  F_{f[c} \nablao_{b]} n^f
=  \nablao_{[c} \sone_{b]} 
+  F_{f [c} \left ( n^f Q_{b]} + P^{fa} (U_{b]a} - \nn F_{b] a} ) \right ) \\
& \quad \quad \quad \quad  \quad
=  \nablao_{[c} \sone_{b]} 
+  \sone_{[c} Q_{b]} +  P^{fa} F_{f[c} U_{b]a} \\
& \quad \quad \quad \quad  \quad 
=  \nablao_{[c} \sone_{b]} -  \nn \sone_{[c} \nablao_{b]} \ll 
+ P^{fa} F_{f[c} U_{b]a}  
\end{align*}
which inserted into (\ref{nB}) yields, after simplification,
\begin{align*}
n^f {\mathcal C}_{fcb} = 
-\frac{1}{2} \sone_{[c} \left ( 
\nablao_{b]} \nn + \nn{}^2 \nablao_{b]} \ll \right ),
\end{align*}
and (\ref{inter}) becomes (\ref{fourth}) once the explicit form of
${\mathcal B}_{cb}$ is used.

To show (\ref{fifth}) we again have two possible routes,
namely 
contract (\ref{second}) with $n^a$ or (\ref{fourth}) with $n^c$.
The fact that both agree is again a non-trivial check.
The second route turns out to be much simpler, so we take this.
Given that $\sone_{c} n^c =0$ one has, from the Cartan identity,
$2 n^c \nabla_{[c} \sone_{b]} = \pounds_{n} \sone_b$ so that the contraction
of (\ref{fourth}) with $n^c$ is
\begin{align*}
\ell_d \Riemoin{}^d{}_{acb} n^a n^c = &  \left ( 1- \ll \nn \right )
\left ( \pounds_{n} \sone_b 
- P^{af} \sone_a U_{bf} - P^{af} F_{ab} U_{cf} n^c + \frac{1}{2}
n( \ll) \nablao_b \nn \right . \\
& \quad \quad \quad \quad \quad \quad  \left . - \frac{1}{2} n(\nn) \nablao_{b} \ll \right )
- \frac{1}{2} \sone_b \left ( n(\nn) + \nn{}^2 n(\ll) \right )
\end{align*}
which is (\ref{fifth}) after inserting (\ref{Un}).

\end{proof}

As shown in the previous section, at non-null points
the connection $\nablao$ can be related to the Levi-Civita connection of 
$\gamma$. Thus, the curvature tensor of $\nablao$ can also
be related to the Riemann tensor of $\gamma$, which we
denote by $\Riemgam$.
\begin{proposition}
Let $\metdata$ by metric hypersurface data and $p \in \N$ a non-null
point. Consider a neighbourhood $U_p$ of $p$ consisting of non-null
points. Then
\begin{align*}
\Riemoin{}^{a}{}_{bcd} 
= \Riemgam{}^{a}{}_{bcd} & +
\ell^a  \left ( 
f \nabgam_c \Pi_{bd} + \Pi_{bd} \left ( \nabgam_c f 
+ \frac{1}{2} f^2 \nabgam_c \ellell 
+ f^2 \acc_c  \right) 
\right . \\
& \left .
\quad \quad - f \nabgam_d \Pi_{bc} - \Pi_{bc} \left ( \nabgam_d f 
+  \frac{1}{2} f^2 \nabgam_d \ellell + f^2 \acc_d 
\right) \right ) \\
& + \frac{f}{2}  \left ( \Pi^{a}{}_c \Pi_{bd} - \Pi^a{}_{d} \Pi_{bc} \right )
+ f \left ( F^a{}_d \Pi_{bc} - F^{a}{}_{c} \Pi_{bd} \right ),
\end{align*}
where indices are lowered an raised using $\gamma$ and $\gamma^{\sharp}$,
$\acc := \nabgam_{\ell^{\sharp}} \ell^{\sharp}$ is the acceleration 
of $\ell^{\sharp}$,  $\ellell := \gamma(\ell^{\sharp},\ell^{\sharp})$,
$f:= - \frac{1}{2 ( \ll - \ellell)}$ and
\begin{align*}
\Pi_{ab} := (\pounds_{\ell^{\sharp}} \gamma)_{ab} =
\nabgam_a \ell_b + \nabgam_b \ell_a 
\end{align*}
is the deformation tensor of $\gamma$ along $\ell^{\sharp}$.
\end{proposition}

\begin{proof}
We use the standard relation (see e.g. \cite{Wald}) between the curvature tensors
of two torsion-free connections $\nabla^{(2)}$ and $\nabla^{(1)}$
with difference tensor $S:= \nabla^{(2)} - \nabla^{(1)}$
\begin{align}
R^{(2)}{}^{a}{}_{bcd} = 
R^{(1)}{}^a{}_{bcd} + \nabla^{(1)}_c S^{a}{}_{bd}
- \nabla^{(1)}_d S^{a}{}_{bc} + S^{a}{}_{cf} S^{f}{}_{bd}
- S^{a}{}_{df} S^{f}{}_{bc}.
\label{IdenCurv}
\end{align}
In our case, with $\nabla^{(1)} = \nabgam$ and
$\nabla^{(2)} = \nablao $ the difference tensor
is given in Proposition \ref{DifCon}, namely
\begin{align*}
S^{a}{}_{bd}   = -\frac{1}{2 ( \ll - \gamma(\ell^{\sharp},
\ell^{\sharp}))}  \ell^a \pounds_{\ell^{\sharp}} \gamma_{bd}=
f \ell^a \Pi_{bd}.
\end{align*}
Recall the definition of $F = \frac{1}{2} d \ellc$ so that
\begin{align*}
\nabgam_c \ell_a = 
\frac{1}{2} \left ( \nabgam_{c} \ell_a - \nabgam_a \ell_c \right )
+ \frac{1}{2} \left ( \nabgam_{c} \ell_a + \nabgam_a \ell_c \right )
= F_{ca} + \frac{1}{2} \Pi_{ac}
\end{align*}
and we obtain
\begin{align}
\nabgam_{c} S^a{}_{bd} = 
(\nabgam_c f ) \ell^a \Pi_{bd} + f
\left ( F_c{}^a + \frac{1}{2} \Pi_{c}{}^a \right ) \Pi_{bd} 
+ f \ell^a \nabgam_c \Pi_{bd}.
\label{firstStep}
\end{align}
On the other hand
\begin{align*}
S^{a}{}_{cf} \ell^f = 
f \ell^a \Pi_{cf} \ell^f =
f \ell^a \left ( \nabgam_c \ell_f + \nabgam_f \ell_c \right ) \ell^f =
f \ell^a \left ( \frac{1}{2} \nabgam_c \ellell + \acc_c \right )
\end{align*}
so that
\begin{align}
S^a{}_{cf} S^f{}_{bd} = f \Pi_{bd} S^{a}{}_{cf} \ell^f 
= f^2 \ell^a \Pi_{bd} 
\left ( \frac{1}{2} \nabgam_c \ellell + \acc_c \right ).
\label{secondStep}
\end{align}
Inserting (\ref{firstStep}) and (\ref{secondStep}) into (\ref{IdenCurv}) yields the result at once.
\end{proof}

\section{Matter-hypersurface data and constraint
equations}
\label{Sect:Matter}

In this section we summarize the main result in \cite{MarsGRG}, as this
plays a fundamental role in the following subsection.
We start by recalling the notion of matter-hypersurface data
\cite{MarsGRG}.

\begin{definition}[Matter-Hypersurface data]
\label{MatterHyp}
A seven-tuple $\mathypdata$ formed by
 hypersurface data $\{\N,\gamma$, $\ellc$, $\ll, \Y  \}$,
a scalar  $\rho_{\ell}$ and a one-form $\bm{J}$ on $\N$
is called {\bf matter-hypersurface data} provided $\rho_{\ell}$
and $\bm{J}$ have the gauge behaviour
\begin{eqnarray*}
\G_{(u,\gauge)} (\rho_{\ell}) =  \rho_{\ell} + \bm{J} (\gauge) , \quad \quad
\G_{(u,\gauge)} (\bm{J})  = \frac{1}{u} \bm{J}.
\end{eqnarray*}
and the
following  {\bf constraint field equations} hold.
\begin{align}
\rho_{\ell} & =
 \frac{1}{2} \Riemo{}^{c}_{\,\,\,bcd} P^{bd} + \frac{1}{2}
\ell_a \Riemo{}^{a}_{\,\,\,bcd} P^{bd} n^c+
\nablao_d \left ( \left ( P^{bd} n^c - P^{bc} n^d \right ) \Y_{bc}  \right )
 + \frac{1}{2} \nn P^{bd} P^{ac} \left ( \Y_{bc} \Y_{da} - \Y_{bd} \Y_{ca} \right )  
 \nonumber  \\
& +\frac{1}{2} \left (P^{bd} n^c - P^{bc} n^d \right )
\left [ \ll \nablao_d \U_{bc}+
   \left ( \U_{bc} + \nn \Y_{bc}  \right ) \nablao_d \ll 
+ 2 \Y_{bc} \left ( F_{df} - \Y_{df} \right ) n^f
\right ],
\label{rhol}  \\
J_c  & = 
   \ell_a \Riemo{}^a_{\,\,\,bcd} n^b n^d
- \nablao_f \left [ \left ( \nn P^{bd} - n^b n^d \right ) \left (
\delta^f_{d} \Y_{bc}  - \delta^f_c \Y_{bd} \right ) \right ] - 
 \left ( P^{bd} - \ll n^b n^d \right ) \left ( \nablao_d \U_{bc} -
\nablao_c \U_{bd} \right ) \nonumber \\ 
& 
- \left ( \nn P^{bd} - n^b n^d \right ) \left [
 \frac{1}{2} \left ( \U_{bc} + \nn \Y_{bc} \right ) \nablao_d \ll - \frac{1}{2}
\left ( \U_{bd} + \nn \Y_{bd} \right ) \nablao_c \ll 
+ \left ( \Y_{bc} F_{df}  - \Y_{bd} F_{cf} \right ) n^f \right ] \nonumber \\
& - \left ( P^{bd} n^f - P^{bf} n^d \right ) \Y_{bd} \U_{cf}
- P^{bd} n^f \left ( \U_{bc} F_{df} - \U_{bd} F_{cf} \right ).
\label{Jc} 
\end{align}
\end{definition}

The motivation behind this definition lies in the following result 
\cite{MarsGRG}, which in essence states that
the
notions of matter-hypersurface data and embedded
hypersurface data are fully compatible.

\begin{theorem}
\label{teor2}
Let $\{ \N,\gamma,\ellc,\ll,{\bm Y}, \rho_{\ell}, \bm{J}\}$ be
matter-hypersurface data and assume that
$\{ \N,\gamma,\ellc,\ll,{\bm Y}\}$ is
embedded  with embedding
$\Phi$, rigging $\rigging$ and ambient spacetime $(\M,\gM)$. Denote by 
$\Ein_g$ the Einstein tensor of $(\M,\gM)$. 

Let $\normal$ be the (unique) normal vector of $\Phi(\N)$
satisfying $g(\normal,\rigging) =1$ and $X \in \X(\N)$ any vector field.
 Then the following identities hold
\begin{align*}
& \Ein_g(\normal, \rigging ) |_{\Phi(\N)} = - \rho_{\ell},  \\
& \Ein_g( \normal,  \Phi_{\star}(X) ) |_{\Phi(\N)}   = -\bm{J} (X). 
\end{align*}
\end{theorem}

\section{Spherically symmetric hypersurface data}
\label{Sect:Spher}

In the remaining of the paper 
we assume that the hypersurface data is spherically
symmetric. i.e. that there is an action 
$\Psi: SO(m) \times \N \longrightarrow  \N$ with 
$(m-1)$-dimensional orbits (which may degenerate to points) and
satisfying
\begin{align*}
\Psi_\a^{\star} (\gamma) = \gamma, \quad \quad
\Psi_\a^{\star}(\ellc)= \ellc,
\quad \quad
\Psi^{\star}_\a(\ll) = \ll,
\quad \quad \Psi^{\star}_\a(\Y) = \Y,
\end{align*}
where $\Psi_\a $ is the diffeomorphism $\Psi_\a := \Psi(\a, \cdot) : 
\N \rightarrow \N$, $\a \in SO(m)$. We restrict to $m>2$, as 
the case $m=2$ is special since $SO(2)$ is one-dimensional, hence
locally indistinguishable to any other local isometry group of the data. 
Denote by $\P$ the collection of centers of symmetry (i.e. the fixed
points of the $SO(m)$ action) and let $\lambda :
\N \setminus \P \rightarrow \mathbb{R}$ be a defining function
of the orbits of $SO(m)$. We denote by
$I$ the image of $\N \setminus \P$. Connectedness of $\N$ implies that $I$ is an interval.
It is clear that the symmetric
tensor
$\gamma$ can be decomposed as
\begin{align*}
  \gamma = \cN d \lambda^2 + d \lambda \otimes_s \bm{\alpha} (\lambda)
  + T(\lambda),
\end{align*}
where $\alpha(\lambda)$ is a one-form field
on the orbit ${\mathcal O}_{\lambda}$ defined by $\lambda$ and
$T(\lambda)$ is, for each $\lambda$, a symmetric
two-covariant tensor on ${\mathcal O}_{\lambda}$. 
An immediate application of the Schur lemma implies that
$\bm{\alpha}(\lambda)=0$ and that $T(\lambda)$ is necessarily
proportional to the standard metric on the $(m-1)$-dimensional
sphere, which we denote by $\gsph$.  Thus there exist two functions
$\cN, Q \in C^{\infty}(I,\mathbb{R})$ such that
\begin{align*}
  \gamma = (\cN \circ \lambda)  d \lambda^2 + (Q \circ \lambda)
  \gsph
\end{align*}
From now on and for simplicity, we drop the composition with
$\lambda$, so that e.g. we write $\cN$ instead of
$\cN \circ \lambda$. A similar argument for $\ellc$ and $\Y$
shows that the metric hypersurface data
tensors $\gamma, \ellc, \ll, \Y$ can be written
as
\begin{align}
\gamma &= \cN d\lambda^2 + Q \gsph, \quad \quad \quad \quad
\ellc = \ellr d \lambda, \label{sphFormmet}  \\
\Y & = \Yr d \lambda^2 + \YT \gsph \label{sphFormhyp}
\end{align}
where $\cN, Q, \Yr, \YT, \ellr, \ll  \in C^{\infty}(I, \mathbb{R})$.
The tensor $\A$ reads,
in matrix form,
\begin{align}
  \A =
  \left ( \begin{array}{ccc}
              Q \gsph{}_{AB} & 0^T & 0 \\
              0 & \cN & \ellr \\
              0 & \ellr & \ll
                           \end{array}
\right ),
\label{MatASph}
\end{align}
so the condition of non-degenerate requires
both $Q \neq 0$ and $\ellr^2 - \cN \ll \neq 0$ everywhere.
The corresponding
tensors $\{ P,\n,\nn\}$ are obtained by solving (\ref{EqP1})-(\ref{EqP4})
(or inverting \ref{MatASph}) 
take the form
\begin{align}
P = \frac{\ll}{\cN \ll - \ellr^2} \partial_{\lambda} \otimes
\partial_{\lambda} + \frac{1}{Q} \gsph^{\sharp}, \quad \quad \n = \frac{\ellr}{\ellr^2
- \ll \cN} \partial_{\lambda}, \quad \quad
\nn = -\frac{\cN}{\ellr^2 - \ll \cN},
\label{inver}
\end{align}
where $\gsph^{\sharp}$ is the contravariant metric associated to
$\gsph$. It is immediate that the signature of $\A$ is (compare
Lemma \ref{signa})
\begin{align*}
  &
  \mbox{sign} (\A|_p) = \{ \sign(\cN), \sign(\cN), \sign(Q), \cdots, \sign(Q) \}
\quad \quad & & & \mbox{if} \quad \quad    \cN \ll - \ellr^2 > 0, \\
& \mbox{sign} (\A|_p) = \{ -1 , +1 , \sign(Q), \cdots, \sign(Q) \}
\quad \quad & & & \mbox{if} \quad \quad   \cN \ll - \ellr^2 < 0.
\end{align*}
Since we are mainly interested in the case with the signature of $\A$
is Lorentzian $\{-,+ \cdots +\}$, we assume from now on that $Q := R^2$
with $R \in C^{\infty}(I,\mathbb{R}^+)$ and
$\cN \ll - \ellr^2 <0$.
The condition that $\ellc$ and $\ll$ are spherically
symmetric reduces the gauge freedom of the data. The residual gauge freedom
$(u,\zeta)$ consists of radially symmetric nowhere zero functions $u(\lambda)$
and vector fields of the form $\zeta = \zetar(\lambda) \partial_{\lambda}$. Under
such a gauge parameter, the functions $\{ \ll, \ellr \}$ transform as
\begin{align}
\G_{(u,\zeta)} (\ellr) = u \left ( \ellr + \cN \zetar \right ), 
\quad \quad
\G_{(u,\zeta)} (\ll) = u^2 \left ( \ll + 2 \ellr \zetar + \cN \zetar{}^2 \right ), 
\label{llprime}
\end{align}
which implies 
\begin{align}
\G_{(u,\zeta)} (\ellr^2 - \cN \ll) = u^2 (\ellr^2 - \cN \ll)
\label{GaugeS2}
\end{align}
as it must be from (\ref{changeA}). It is of interest to introduce gauge 
invariant quantities. This can be accomplished under the additional
condition that we work in a domain where $\ellr \neq 0$ and restrict 
the gauge
transformations to be sufficiently close to the identity. Note that
$\ellr \neq 0$ holds necessarily  at null points (characterized by the
vanishing of $\cN$). In order to construct gauge invariant quantities
we start with the following  result:
\begin{lemma}
\label{GaugeInv}
Let $\metdata$ be spherically symmetric metric hypersurface data, hence
given by (\ref{sphFormmet}).
Assume that $\ellr$ is nowhere zero on a domain $\UU \subset \N$. 
Then the vector field 
\begin{align}
  V := \left \{ \begin{array}{c}
    \frac{1}{\cN} \left ( \ellr - \sign(\ellr) 
    \sqrt{\ellr^2 - \ll \cN } \right ) \partial_{\lambda} \quad \quad N \neq 0 \\
    \frac{1}{2} \frac{\ll}{\ellr} \partial_{\lambda}   \quad \quad N=0
  \end{array}
  \right .
  \label{defV}
\end{align}
transforms as
\begin{align}
\G_{(u,\zeta)}(V) = u ( V + \zeta)
\label{GaugeV}
\end{align}
provided the gauge vector $\zeta = \zetar(\lambda) \partial_{\lambda}$ satisfies
\begin{align}
  \sign(\ellr + \cN \zetar ) = \sign (\ellr).
\label{sign}
\end{align}
\end{lemma}

\begin{proof}
We first show that $V$ is well-defined on $\UU$,
 i.e. independent of the choice of radial coordinate $\lambda$ and smooth across
null points. The most general change of radial coordinate is given by
$\lambda = \lambda (\lambda')$ with $\frac{d \lambda'}{d \lambda} $ nowhere
zero. Under such coordinate change, the coefficients $\cN$ and 
$\ellr$ transform as
\begin{align*}
\cN' = \cN \left ( \frac{ d \lambda}{d \lambda'} \right )^2,
\quad \quad \ellr' = \ellr \frac{d \lambda}{d \lambda'}.
\end{align*}
Moreover
\begin{align}
\sqrt{\ellr{}'^2 - \cN' \ll } = 
\left | \frac{d \lambda}{d \lambda'} \right | 
\sqrt{\ellr{}^2 - \cN \ll } 
\label{TransS}
\end{align}
so that, at points where $N \neq 0$,
the corresponding $V'$, defined by 
\begin{align*}
  V' & := \frac{1}{\cN'} \left ( \ellr' - \sign(\ellr') 
\sqrt{\ellr'{}^2 - \ll \cN' } \right ) \partial_{\lambda'}  \\
& = \frac{1}{\cN} \left ( \frac{d \lambda'}{d \lambda} \right )^2
\left ( \ellr \frac{d \lambda}{d \lambda'}
- \sign(\ellr) \sign \left ( \frac{d \lambda}{d \lambda'} \right )
\left | \frac{d \lambda}{ d \lambda'} \right |
\sqrt{\ellr^2 - \ll \cN } \right ) \partial_{\lambda'} \\
& = 
\frac{1}{\cN} 
\left ( \ellr 
- \sign(\ellr)
\sqrt{\ellr^2 - \ll \cN } \right ) \frac{d \lambda'}{d \lambda}
\partial_{\lambda'}  
\end{align*} 
corresponds exactly to the vector field $V$ expresses in the
new coordinate $\lambda'$. A similar  argument applies to the points where
$N =0$.

Concerning smoothness. the only problematic points
are those where $\cN$ becomes zero. However, a simple Taylor
expansion shows (\ref{defV}) is actually smooth there, hence
everywhere. We now analyze the gauge behaviour of $V$. First  observe that
\begin{align*}
\sign(\G_{(u,\zeta)} (\ellr) ) = \sign \left
( u \left ( \ellr + \cN \zetar \right ) \right ) = \sign (u) \sign(\ellr)
\end{align*}
so that, using (\ref{GaugeS2}),
\begin{align*}
\G_{(u,\zeta)} (V) = \frac{1}{\cN} \left ( u ( \ellr + \cN \zetar )
 - \sign(u) \sign(\ellr) |u| \sqrt{\ellr^2 - \ll \cN} \right ) \partial_{\lambda}
= u V + u \zetar \partial_{\lambda}
\end{align*}
and (\ref{GaugeV}) is established at points where $N\neq 0$.
A similar (and simpler) argument applies at points where $N=0$.
\end{proof}

The vector field $V$ allows us to introduce quantities with very simple
gauge behaviour.  We first define the following
\begin{align*}
  \S := \sign(\ellr) \sqrt{\ellr^2 - \cN \ll}
\end{align*}
which by (\ref{GaugeS2}) satisfies 
$\G_{(u,\zeta)} (\S) = u \S$ whenever
 $(u,\zeta)$ is restricted to fulfill (\ref{sign})

\begin{lemma}
\label{GaugeObjects}
Under the hypotheses of Lemma \ref{GaugeInv}, the symmetric tensor
$\Z:=\frac{1}{2} \pounds_{V} \gamma- Y$ transforms under a gauge parameter $(u,\zeta)$
restricted to satisfy (\ref{sign}) as
\begin{align}
\G_{(u,\zeta)} (\Z) =
u \Z - \S du \otimes_s d \lambda.
\label{YLieg}
\end{align}
Decomposing this tensor as
\begin{align*}
\Z = \Zr d\lambda^2 + \ZT \gsph,
\end{align*}
the quantities
\begin{align*}
\Gamma := \frac{1}{S} \left ( \frac{dS}{d\lambda} + \Zr \right ), \quad \quad \quad 
h := \frac{\ZT}{\S},
\end{align*}
are gauge invariant, i.e. satisfy
$\G_{(u,\zeta)}(\Gamma) = \Gamma$ and $\G_{(u,\zeta)} (\ZT ) =  \ZT$.
Moreover, under a change of coordinates $\lambda(\lambda')$, they
transform as
\begin{align}
\Gamma' = \Gamma \frac{d \lambda}{d \lambda'} + \frac{d^2 \lambda}{d 
\lambda'{}^2} \frac{d \lambda'}{d \lambda}, \quad 
\quad \quad 
h' = h \frac{d \lambda'}{d \lambda}.
\label{weight0}
\end{align}
\end{lemma}

\begin{proof}
We start with (\ref{YLieg}). The gauge behaviour of $\Y$ is given by 
(\ref{gaugetrans}). Thus
\begin{align*}
\G_{(u,\zeta)} \left (\frac{1}{2} \pounds_{V} \gamma - \Y \right ) & =
\frac{1}{2} \pounds_{\G_{(u,\zeta)} (V)} \gamma 
- u \Y - \ellc \otimes_s du - \frac{1}{2} \pounds_{u \zeta} \gamma
\\
& =
\frac{1}{2} \pounds_{u V } \gamma + \frac{1}{2}\pounds_{u \zeta} \gamma 
- u \Y - \ellc \otimes_s du - \frac{1}{2} \pounds_{u \zeta} \gamma \\
& =
u  \left (  
\frac{1}{2} \pounds_{V} \gamma - \Y  \right )
+ du \otimes_s \left ( \gamma(V, \cdot) - \ellc \right )
\end{align*}
where in the last equality we used the identity
(\ref{Lief}).
Expression (\ref{YLieg})  follows
by noting
\begin{align*}
\gamma(V, \cdot) - \ellc =
N V^{\lambda} d \lambda - \ellr d\lambda = 
- \sign(\ellr) \sqrt{\ellr^2 - \ll \cN} d \lambda
= - \S d \lambda.
\end{align*}
From the definition of $\Zr$ and $\ZT$ and the gauge behaviour of $\Z$, it 
follows
\begin{align*}
  \G_{(u,\zeta)} (\Zr) = u \Zr - S \frac{du}{d \lambda}, \quad \quad
\G_{(u,\zeta)} (\ZT) = u \ZT
\end{align*}
so that
\begin{align*}
\G_{(u,\zeta)} (\Gamma) & = \frac{1}{\G_{(u,\zeta)}(S)}
\left (  \frac{d}{d\lambda}
\left ( \G_{(u,\zeta)}(\S) \right ) + \G_{(u,\zeta)}(\Zr) \right ) = 
\frac{1}{u \S} \left ( \frac{d (u \S) }{d\lambda} + u \Zr - \S \frac{du}{d \lambda}
\right) =
\frac{1}{S} \left ( \frac{d \S}{d\lambda} + \Zr \right )  = \Gamma, \\
\G_{(u,\zeta)} (h) & = \frac{\G_{(u,\zeta)}(\ZT)}{\G_{(u,\zeta)}(\S)} 
= \frac{u \ZT}{u \S} = h.
\end{align*}
It only remains to analyze the coordinate behaviour of these objects.
$\ZT$ is obviously
a scalar, while $\S$ transforms as (\ref{TransS}) 
\begin{align*}
\S' = \sign(\ellr') 
\left | \frac{d \lambda}{d \lambda'} \right |  \sqrt{\ellr^2 - \ll \cN} = 
\sign \left ( \frac{d \lambda}{d \lambda'} \ellr \right )
\left | \frac{d \lambda}{d \lambda'} \right |  \sqrt{\ellr^2 - \ll \cN} = 
\frac{d \lambda}{d \lambda'} \S.
\end{align*}
This establishes the second in (\ref{weight0}). For the first
\begin{align*}
\Gamma' = \frac{1}{S'} \left ( \frac{d S'}{d \lambda'} + \Zr' 
\right )
= \frac{1}{\S}  \frac{d \lambda'}{d \lambda} \left ( 
\frac{d}{d \lambda'} \left ( \frac{d \lambda}{d \lambda'}  \S \right )
+ \Zr \left ( \frac{d \lambda}{d \lambda'} \right )^2 \right )
= \frac{d \lambda'}{d \lambda} \frac{d^2 \lambda}{d \lambda'{}^2}
+ \frac{d \lambda}{d \lambda'} \left ( \frac{1}{\S} \frac{d \S}{
d \lambda} + \Zr \right ),
\end{align*}
where we used that $\Z$ is a tensor and hence its $(\lambda, \lambda)$ 
component transforms as $\Zr' = \Zr \left ( \frac{d \lambda}{d \lambda'} 
\right )^2$.
\end{proof}

Thus, from  the tensor field $Z$ we have built two gauge invariant
quantities $\Gamma$ and $h$ which, in addition, have useful transformation
properties under change of radial coordinate. Specifically, as we shall see
next, $\Gamma$
behaves as a connection in the radial direction and this
will  allow us to define a radial covariant
derivative on functions of $\lambda$ with a well-defined {\bf weight}.
\begin{definition}
A quantity $H(\lambda)$ is said to have {\bf weight} $r$, with $r \in
\mathbb{Z}$ provided its transformation under a coordinate
change $\lambda(\lambda')$ is
\begin{align*}
H'(\lambda') =  H(\lambda(\lambda')) \left ( \frac{d \lambda}{d 
\lambda'} \right )^r.
\end{align*}
\end{definition}
Examples of objects with well-defined weight are $\cN$ (weight $2$),
$R$ (weight $0$) or $h$ (weight $-1$).  
The role of $\Gamma$ as a connection is shown in the following result.
\begin{lemma}
\label{defdeltalambda}
Let $H$ be of weight $r$ and
define the derivative
\begin{align*}
\delta_{\lambda} H := \frac{d H}{d \lambda} - r \Gamma H.
\end{align*}
Then $\delta_{\lambda} H$ is of weight $r+1$.
\end{lemma}
\begin{proof}
Consider the coordinate change $\lambda(\lambda')$. Let us compute
$\delta_{\lambda} H$ in the coordinates $\lambda'$:
\begin{align*}
\left ( \delta_{\lambda} H  \right )^{\prime} & = 
\frac{d}{d\lambda'} \left ( H'(\lambda') \right ) - r \Gamma' H'
 = \frac{d}{d \lambda'} \left ( H \left ( \frac{d \lambda}{d \lambda'} \right )^r
\right ) 
- r \left ( \Gamma \frac{d\lambda}{d \lambda'}
+ \frac{d^2 \lambda}{d \lambda'{}^2} \frac{d \lambda'}{d \lambda} \right ) H \left ( \frac{d \lambda}{
d \lambda'} \right )^{r} \\
& = \frac{d H}{d \lambda} \left ( \frac{d \lambda}{d \lambda'} \right )^{r+1}
+ r H 
\left ( \frac{d \lambda}{d \lambda'} \right )^{r-1}
\frac{d^2 \lambda}{d \lambda'{}^2} 
- r \Gamma H \left ( \frac{d \lambda}{d\lambda'} \right )^{r+1}
- r H \left ( \frac{d \lambda}{d \lambda'} \right )^{r-1}
\frac{d^2 \lambda}{d \lambda'{}^2}  \\
& = \left ( \frac{d H}{d\lambda} - r \Gamma H \right ) 
\left ( \frac{d \lambda}{d \lambda'} \right )^{r+1}
= \left ( \frac{d \lambda}{d \lambda'} \right )^{r+1}
\delta_{\lambda} H
\end{align*}
as claimed.
\end{proof}

\begin{remark}
Note that if $H_1$ is of weight $r$ and $H_2$ is of weight $s$
then $H_1 H_2$ has weight $r+s$ and the Leibniz rule holds
\begin{align*}
\delta_{\lambda} (H_1 H_2) = H_1 \delta_{\lambda} H_2 + H_2 \delta_{\lambda} H_1.
\end{align*}
\end{remark}

The covariant derivative $\delta_{\lambda}$ defines a flat connection:
\begin{lemma}
\label{flat}
For any spherically symmetric hypersurface data, there always exists a choice
of  radial coordinate $\lambda_a$ such that $\delta_{\lambda} = \frac{d}{d\lambda_a}$
on objects of arbitrary weight. Moreover $\lambda_a$ is unique up to affine
transformations.
\end{lemma}

\begin{proof}
We start with a coordinate $\lambda'$ with a possibly non-zero
$\Gamma'(\lambda')$. We need to show that there exists an invertible
$\lambda_a(\lambda')$ such that $\Gamma=0$. From the transformation law
for $\Gamma$ (\ref{weight0}) this 
is equivalent to 
\begin{align*}
\frac{d \lambda_a}{d \lambda'} \Gamma'(\lambda') = 
\frac{d^2 \lambda_a}{d \lambda'{}^2}
\end{align*}
Existence of solutions is obvious.  It is also clear that
the general (invertible) solution is
$\lambda_a = a f(\lambda') + b$ with $a,b \in \mathbb{R}$, 
$a \neq 0$  and 
$f(\lambda')$ a particular solution with nowhere zero derivative.
\end{proof}

The vector field $V$ also allows one to construct a gauge invariant matter 
hypersurface quantity. 
\begin{lemma}
\label{rhoinv}
Let $\mathypdata$  be matter-hypersurface
data. Under the assumptions of Lemma \ref{GaugeInv} the quantity
\begin{align*}
\tilrho:= \rho_{\ell} - \bm{J}(V)
\end{align*}
is gauge invariant.
\end{lemma}

\begin{proof}
Applying the gauge behaviour of $\rho_{\ell}$ and $\bm{J}$ (\ref{MatterHyp})
\begin{align*}
\G_{(u,\zeta)} (\rho_{\ell} - \bm{J}(V)) = 
\rho_{\ell} + \bm{J}(\zeta) - \frac{1}{u} \bm{J} ( \G_{(u,\zeta)}(V)) =
\rho_{\ell} + \bm{J}(\zeta) - \frac{1}{u} \bm{J} ( u(V + \zeta)) = 
\rho_{\ell} - \bm{J}(V).
\end{align*}
\end{proof}
We have now the necessary ingredients to write down the constraint
equations for spherically symmetric hypersurface data. We have
introduced a number of quantities which are either gauge invariant
($\Gamma$ and $\Z_T$) or  have  
simple scaling behaviour with $u$ ($\S$ and $\bm{J}$) 
under gauge transformations
$(u,\zeta)$. We call this family of five quantities  {\bf gauge covariant}.
Our aim is to write down the constraint field equations in terms
of these objects. Even more, $\Gamma$ being a radial connection we want to write
down equations involving only quantities with
well defined weights. In particular all
radial derivatives should involve the $\delta_{\lambda}$ operator.
The computations could be done without imposing any gauge
or, alternatively and in a much simpler way,
we may choose first a convenient gauge  to perform the
intermediate computations and
rewrite the final result in terms
of covariant quantities in such a way that we recover the 
explicitly obtained  expression
in the selected gauge.  The gauge covariance of the final
result implies then its validity in any gauge.

It turns out that a specific gauge that  simplifies greatly the computations
is the one satisfying $\ll =0$. We first make sure that under
the assumptions of Lemma \ref{GaugeInv} this gauge exists and it
is in fact unique
up to scaling.
\begin{lemma}
\label{gaugechoice}
Given spherically symmetric metric hypersurface data $\metdata$ with 
$\ellr \neq 0$, there exists a spherically symmetric gauge parameter $(u_0,\zeta_0)$
satisfying (\ref{sign})
such that 
$\G_{(u,\zeta)} (\ll) = 0$. In fact, the most general such gauge is given by
$(u,\zeta_0)$ where $u$ an arbitrary nowhere zero
spherically symmetric function.
\end{lemma}
\begin{proof}
From (\ref{llprime}) the condition $\G_{(u,\zeta)} (\ll)=0$ has solutions 
\begin{align*}
\zetar &= \frac{1}{\cN} \left ( - \ellr + \epsilon \sqrt{\ellr^2 - \ll \cN}
\right ) \quad \quad \quad & \mbox{if} \quad \quad \cN \neq 0 \\
\zetar & = - \frac{\ll}{2 \ellr} & \mbox{if} \quad \quad \cN =0
\end{align*}
where $\epsilon =\pm 1$. The sign condition (\ref{sign}) forces 
$\epsilon = \sign (\ellr)$. So $\zetar$ exists and is unique, while $u$
is arbitrary. This proves the result.
\end{proof}

Note that in the gauge defined by $(u,\zeta_0)$, the vector field $V$ vanishes.
So in this specific gauge $\tilrho = \rho_{\ell}$.

We now proceed with the computation of the right-hand sides of the constraint
field equations in the  gauge constructed in Lemma \ref{gaugechoice}.
First of all, spherical symmetry implies that (in the adapted
coordinates $\{ \lambda, \omega^A\}$, where $\omega$ are coordinates
in $\mathbb{S}^{m-1}$), it must be $\bm{J}_A =0$. Also the right-hand
side of (\ref{Jc}) with $c=A$ is identically vanishing as a consequence
of the symmetry. Thus, we only need to compute the right-hand side
of (\ref{rhol}) and of the component $c= \lambda$ of (\ref{Jc}).

We start with the curvature terms. Observe (\ref{inver})
$P = \frac{1}{R^2} \gsph^{\sharp}$, so
\begin{align*}
\Riemo{}^{c}{}_{bcd} P^{bd} + \ell_a \Riemo{}^{a}{}_{bcd} P^{bd} n^c & =
\frac{1}{R^2} \left ( \Riemo{}^{c}{}_{BcD} 
+ \ellr \nr \Riemo{}^{\lambda}{}_{B \lambda D} \right ) \gsph^{BD} \\
& =  
\frac{1}{R^2} \left ( 2 \Riemo{}^{\lambda}{}_{B\lambda D} 
+ \Riemo{}^{A}{}_{B A D} \right ) \gsph^{BD}
\end{align*}
where in the third equation we used $\nr\ellr=1$ (\ref{inver}). To compute
these terms we need the connection symbols $\Gamo$ of $\nablao$. 
Given the form of $P$, expression  (\ref{GambPropo})  yields
immediately
\begin{align*}
\Gamo{}^{\lambda}_{ab} & = \nr 
\dot{\ellr}  
\delta^{\lambda}_{a} \delta^{\lambda}_{b} =
\frac{\dot{\ellr}}{\ellr} 
\delta^{\lambda}_{a} \delta^{\lambda}_{b}, \\
\Gamo{}^{C}_{ab} & = \frac{1}{2} P^{CD} \left ( 
\partial_a \gamma_{D b} +
\partial_b \gamma_{D a} -
\partial_D \gamma_{ab} \right )=
 \left ( \delta^{\lambda}_a \delta^{C}_b 
+ \delta^{\lambda}_b \delta^{C}_a \right ) \frac{\dot{R}}{R}
+ \delta^{A}_a \delta^B_b \Gamma_{\sph^{m-1}}{}^C_{AB},
\end{align*}
where dot means derivative with respect to $\lambda$,
and $\Gamma_{\sph^{m-1}}{}^C_{AB}$ are the Christoffel symbols of 
$\gsph$. Thus
\begin{align*}
\Riemo{}^{\lambda}{}_{B\lambda D} &=
\Gamo{}^{\lambda}_{\lambda a} \Gamo{}^{a}_{BD} -
\Gamo{}^{\lambda}_{D a} \Gamo{}^{a}_{B\lambda} 
= 0, \\
\Riemo{}^{A}{}_{BAD} &= 
\partial_A \Gamo{}^{A}_{BD} - 
\partial_D \Gamo{}^{A}_{BA} + \Gamo{}^A_{A\lambda} \Gamo{}^{\lambda}_{BD}
+ \Gamo{}^A_{A C} \Gamo{}^{C}_{BD}
- \Gamo{}^A_{D\lambda} \Gamo{}^{\lambda}_{BA} 
- \Gamo{}^A_{A C} \Gamo{}^{C}_{BD} \\
& = \Riem_{\sph^{m-1}}{}^{C}_{BCD} = (m -2) \gsph{}_{BD}
\end{align*}
where  $\Riem_{\sph^{m-1}}$ is the Riemann tensor of the standard
$m-1$ dimensional sphere. We conclude
\begin{align*}
\Riemo{}^{c}{}_{bcd} P^{bd} + \ell_a \Riemo{}^{a}{}_{bcd} P^{bd} n^c & =
\frac{(m-1)(m-2)}{R^2}.
\end{align*}
Concerning the curvature term in the right-hand side
of the equation for $\bm{J}_{\lambda}$, this is obviously zero 
\begin{align*}
\ell_a \Riemo{}^{a}_{b \lambda d} n^b n^d = 0
\end{align*}
since $\n$ only has component along $\partial_{\lambda}$. 

The one form $\ellc$ is closed, so $F_{ab}=0$. Moreover,
$P^{bd} n^c \Y_{bc}= 0$ and
$P^{bc} \Y_{bc} = \frac{(m-1) \YT}{R^2}$ follow
immediately from the structure of $P$, $n$ and the form of $\Y$
(cf. (\ref{sphFormhyp})). Using also that $\ll=0$, 
the constraint equation (\ref{rhol}) becomes
\begin{align}
\rho_{\ell} &= 
\frac{(m-1)(m-2)}{2 R^2}
- \nablao_d \left (  \frac{(m-1) \YT}{R^2} n^d \right )
+ \frac{1}{2} \nn P^{BD} P^{AC} \left ( \Y_{BC} \Y_{DA} - \Y_{BD} 
\Y_{CA} \right ) 
+ P^{bc} \Y_{bc} \Y_{df} n^d n^f
\nonumber 
\\
& = 
\frac{(m-1)(m-2)}{2 R^2}
- \nablao_d \left (  \frac{(m-1) \YT}{R^2} n^d \right )
- \frac{\YT^2 \nn}{2 R^4} ( m-1 ) (m-2) + \frac{(m-1) \YT}{R^2} (\nr)^2 \Yr.
\label{eqrhol} 
\end{align}
The
$\nablao$-divergence of a vector $\zeta = H(\lambda) \partial_{\lambda}$
is
\begin{align}
\nablao_d \left ( H (\partial_{\lambda})^d \right ) & =
\dot{H} + H \Gamo{}^d_{da} (\partial_{\lambda})^a =
\dot{H} + H \Gamo{}^{d}_{d\lambda} = 
\dot{H} + H \Gamo{}^{\lambda}_{\lambda\lambda} 
+ H \Gamo{}^{A}_{A \lambda}  = \dot{H} + H \frac{\dot{\ellr}}{\ellr} 
+ H (m-1) \frac{\dot{R}}{R} \nonumber \\
& = 
\frac{\dot{ (H R^{m-1})}}{R^{m-1}} + H \frac{\dot{\ellr}}{\ellr} 
\label{divergence}
\end{align}
so that equation (\ref{eqrhol}) 
can be rewritten
as
\begin{align}
\frac{\rho_{\ell}}{m-1}
= \frac{m-2}{2 R^2} - \frac{1}{R^{m-1}}\frac{d}{d\lambda} \left ( \frac{\Y_T R^{m-3}}{\ellr} 
\right ) - \frac{\YT}{R^2} \frac{\dot{\ellr}}{\ellr^2}
+ \frac{m-2}{2} \frac{\YT^2 \cN}{\ellr^2 R^4}
+ \frac{\YT \Yr}{\ellr^2 R^2}
\label{eqrhol2}
\end{align}
once the explicit form of $\nn$ and $\nr$ from (\ref{inver}) are substituted.
In the present gauge, the vector field $V$ vanishes identically 
(see (\ref{defV})) and
hence the gauge covariant quantities
$\S, \Gamma, h, \tilrho$ take the form
\begin{align}
\label{nullgauge}
\S = \ellr, \quad \quad 
\Gamma = \frac{1}{\ellr} \left ( \frac{d \ellr}{d \lambda} - \Yr
\right ), \quad \quad
h = - \frac{\YT}{\ellr}, \quad \quad
\tilrho = \rho_{\ell}.
\end{align}
Thus, using the fact that
$h R^{m-3}$ has weight $r=-1$,
equation (\ref{eqrhol2}) can be written in terms of gauge covariant
quantities as
\begin{align*}
\frac{\tilrho}{m-1} = \frac{m-2}{2 R^2}
+ \frac{1}{R^{m-1}} \delta_{\lambda} \left ( h R^{m-3} \right )+  \frac{m-2}{2} 
\frac{h^2 \cN}{R^4}.
\end{align*}
Note that all terms
in this expression has the same weight, as required by consistency.

We next turn our attention to the constraint equation for $J_{\lambda}$. Given 
that $\ll=0$, $F_{ab}=0$ and $\ell_a \Riemo{}^{a}{}_{bc d} n^b n^d=0$, 
equation (\ref{Jc}) reduces immediately to
\begin{align*}
J_{c} & = - \nablao_f \left [ 
\left ( \nn P^{bd} - n^b n^d \right ) \left (
\delta^f_{d} \Y_{b c}  - \delta^f_c \Y_{bd} \right ) 
\right ] - P^{bd} 
\left ( \nablao_d \U_{bc} -
\nablao_c \U_{bd} \right )  - \left ( P^{bd} n^f - P^{bf} n^d \right ) \Y_{bd} \U_{c f}.
\end{align*}
We need the explicit form of  $\U_{ab}$. First, we use 
Lemma \ref{SomeIden} with $\bm{\sone}  =  i_{\n} F =0$ and $\ll=0$
to conclude
\begin{align}
  \U(n,\cdot) = \frac{1}{2} d \nn \quad \quad
\Longrightarrow \quad \quad P^{ab} \U_{ac} n^c =0, \quad \U(n,n)  = \frac{1}{2 \ellr} \frac{d \nn}{d \lambda}. \label{Uncdot}
\end{align}
The $AB$ components of $U_{ab}$ are immediate from its definition
(\ref{defU}):
\begin{align*}
\U_{AB} = \frac{R \dot{R}}{\ellr} \gsph{}_{AB}.
\end{align*}
We are only concerned with $\bm{J}(n)$,
so we contract (\ref{Jc}) with $n^c$ and 
``integrate by parts'' to get
\begin{align*}
\bm{J} (\n) = & 
- \nablao_{f} \left [
\left ( \nn P^{bd} - n^b n^d \right ) \left (
\delta^f_{d} \Y_{b c}  - \delta^f_c \Y_{bd} \right ) n^c \right ]
+ \left ( \nn P^{bd} - n^b n^d \right ) \left (
\delta^f_{d} \Y_{b c}  - \delta^f_c \Y_{bd} \right ) \nablao_f n^c \\
& 
- P^{bd} \nablao_d \left (  \U_{bc} n^c \right ) 
+ P^{bd} \U_{bc} \nablao_d n^c 
+ P^{bd} n^c \nablao_c (\U_{bd} )
- (P^{bd} \Y_{bd} ) (\U_{cf} n^c n^f)
+ P^{bf} n^d \Y_{bd} \U_{cf} n^c.
\end{align*}
Since (cf. Lemma \ref{identitiesnablao}) $P$  
is $\nablao$-covariantly constant  and $\nablao_f n^c = P^{ca} \U_{af}$,
 this expression simplifies to
\begin{align*}
\bm{J} (\n) & =  \nablao_f \left ( \nn P^{bd} \Y_{bd} n^f \right )
+ \left (\nn P^{bd} - n^b n^d \right ) \left (\Y_{bc} P^{ca} \U_{ad}
- \Y_{bd} P^{ca} \U_{ca} \right ) + P^{bd} \U_{bc} P^{ca} \U_{ad} \\
& + n^{c} \nablao_c \left ( P^{bd} \U_{bd} \right ) 
- (P^{bd} \Y_{bd} ) (\U_{cf} n^c n^f) \\
& = \nablao_f \left ( \nn P^{bd} \Y_{bd} n^f \right ) + n^{c} \nablao_c \left ( P^{bd} \U_{bd} \right ) 
+ P^{bd} P^{ca} \left ( \nn \Y_{bc} + \U_{bc} \right ) \U_{da} - \nn (P^{bd} \Y_{bd} ) (P^{ca} \U_{ca} ) \\
& + Y(n,n) P^{ca} \U_{ca}  - \U(n,n) P^{bd} \Y_{bd}  
\end{align*}
 where we used $P^{bf} n^d \Y_{bd} =0$ and
$P^{bf} n^d \U_{bd} =0$. Inserting the form of $\Y$, $\U$ and $P$ yields
immediately
\begin{align*}
\bm{J} (n)  = \,  &
\nablao_{f} \left ( (m-1) \frac{\nn \YT}{R^2} n^f \right )
+ \frac{1}{\ellr} \frac{d}{d\lambda} \left ( (m-1) \frac{\dot{R}}{R \ellr} 
\right ) + (m-1) \frac{\dot{R}}{R^3 \ellr} \left ( \nn \YT + \frac{R \dot{R}}{
\ellr} \right ) \\
& - \nn (m-1)^2 \frac{ \YT \dot{R}}{ \ellr R^3}
+ (m-1) \frac{\Yr \dot{R}}{\ellr^3 R} - \frac{(m-1)}{2} \frac{\YT}{R^2 \ellr}
\frac{d \nn}{d\lambda} \\
 = \, &  \frac{(m-1)}{R^{m-1} \ellr}  
\frac{d}{d\lambda} \left ( \nn \YT R^{m-3}  \right )
+ \frac{m-1}{R \ellr^2} \left ( \frac{d}{d \lambda} \dot{R}
-  \dot{R} \frac{\dot{\ellr}}{\ellr} + \dot{R} \frac{\Yr}{\ellr} 
\right )
- (m-1) (m-2) \frac{\nn \dot{R} \YT}{R^3 \ellr}  \\
& - \frac{m-1}{2} \frac{\YT}{R^2 \ellr} \frac{d \nn}{d\lambda}
\end{align*}
where in the second equality we used expression (\ref{divergence}) for the divergence term. $R$ has weight zero, so $\delta_\lambda R = \dot{R}$
is a weight one quantity. Thus, the term inside the parenthesis in the last expression 
equals $\delta_{\lambda} \delta_{\lambda} R$. Inserting $\nn = -\cN/\ellr^2$ 
and $\YT = - h \ellr$ yields
\begin{align*}
\bm{J}(\n) = (m-1) \left [
 \frac{1}{R^{m-1} \ellr} \frac{d}{d\lambda} \left ( 
\frac{\cN h R^{m-3}}{\ellr} \right ) - (m-2)  \frac{\cN h \dot{R}}{R^3 
\ellr^2 } - \frac{1}{2} \frac{h}{R^2}  \frac{d}{d \lambda} \left ( \frac{\cN}{\ellr^2} \right ) + \frac{1}{R \ellr^2} \delta_{\lambda} \delta_{\lambda} R
\right ].
\end{align*}
Expanding the first derivative we arrive at
\begin{align*}
\frac{\bm{J}(\n)}{m-1} 
= \frac{1}{R^2 \ellr^2} \left ( \frac{1}{2} h \dot{\cN} 
+ \cN \dot{h} - \frac{\cN h}{R} \dot{R} + R
 \delta_{\lambda} \delta_{\lambda} R \right ).
\end{align*}
Now, $R$ is of weight $0$, $N$ is weight $2$ and $h$ is weight $-1$, 
so we may replace all derivatives with respect to $\lambda$
by $\delta_{\lambda}$ derivatives, as all the
connection terms one introduces along the way cancel each other. The final result is
\begin{align*}
\frac{\ellr^2 \bm{J}(\n)}{m-1} =
\frac{1}{2R^2} h \delta_{\lambda} \cN
+ \frac{\cN}{R^2} \delta_{\lambda} h - \frac{\cN h}{R^3} \delta_{\lambda} R + 
\frac{1}{R} \delta_{\lambda} \delta_{\lambda} R.
\end{align*}
Note that all terms in the right hand side are of weight $2$ and gauge
invariant. The right hand side is not yet gauge invariant, but it is immediate
to find a gauge invariant expression that reduces to this in the gauge
$\ll=0$.  Indeed, the quantity $\S \bm{J} (\partial_{\lambda})$ reduces
to $\ellr^2 \bm{J} (\n)$ in this specific gauge and it
is gauge invariant under any gauge satisfying (\ref{sign}):
\begin{align*}
\G_{(u,\zeta)} (\S \bm{J} (\partial_{\lambda})) = u \S \frac{1}{u} \bm{J} \left ( 
\partial_{\lambda} \right ) = S \bm{J} (\partial_{\lambda}).
\end{align*}
We also note that both $\S$ and $\bm{J} (\partial_{\lambda})$
are of weight $1$, so their product is weight two.  
Thus, the gauge invariant and radial-coordinate
covariant equation is
\begin{align*}
\frac{\S \bm{J}(\partial_{\lambda})}{m-1} =
\frac{h}{2R^2}  \delta_{\lambda} \cN
+ \frac{\cN}{R^2} \delta_{\lambda} h - \frac{\cN h}{R^3} \delta_{\lambda} R + 
\frac{1}{R} \delta_{\lambda} \delta_{\lambda} R.
\end{align*}
We can summarize the results of this section in the following theorem.
\begin{theorem}
Let $\{\N^{m},\gamma,\ellc,\ll,\Y,\rho_{\ell},\bm{J} \}$ ($m > 2$)
be
spherically symmetric matter hypersurface data
of Lorentzian ambient signature. Decompose
$\gamma$ as (\ref{sphFormmet}) and define the gauge covariant quantities $\S$, $h$ and
$\Gamma$  as in Lemma \ref{GaugeObjects}, $\tilrho$ as in Lemma
\ref{rhoinv} and the operator $\delta_{\lambda}$ as in Lemma \ref{defdeltalambda}. Then, the 
constraint equations take the following form in any gauge where
$\ellc \neq 0$
\begin{align}
\frac{\tilrho}{m-1} & = \frac{m-2}{2 R^2} \left ( 1 
+ \frac{\cN h^2 }{R^2} \right )
+ \frac{1}{R^{m-1}} \delta_{\lambda} \left ( h R^{m-3} \right ),  \label{firstsph} \\
\frac{\S \bm{J}(\partial_{\lambda})}{m-1} & =
\frac{h}{2R^2}  \delta_{\lambda} \cN
+ \frac{\cN}{R^2} \delta_{\lambda} h - \frac{\cN h}{R^3} \delta_{\lambda} R + 
\frac{1}{R} \delta_{\lambda} \delta_{\lambda} R. \label{secondsph} 
\end{align}
\end{theorem}
We  note for later use
that the second equation in the theorem
implies
\begin{align*}
\frac{1}{m-1} h \S \bm{J} (\partial_{\lambda}) = 
\frac{1}{2} \delta_{\lambda} \left ( \frac{ N h^2}{R^2} \right )
+ \frac{h}{R} \delta_{\lambda} \delta_{\lambda} R. 
\end{align*}
In fact, this equations is equivalent to 
(\ref{secondsph})  whenever $h \neq 0$).
It is thus natural to define the quantity $\beta := \frac{N h^2}{R^2}$,
which is both gauge invariant
and of weight zero.

\section{Birkhoff theorem at the abstract hypersurface
data level}
\label{Sect:Birkhoff}

In this subsection we want to find the most general solution of 
the constraint equations for Lorentzian ambient
signature in the {\it vacuum} case, by which we 
mean $\tilrho=0$ and $\bm{J}=0$. 
Since we are assuming $m>2$, equation  (\ref{firstsph}) with
$\tilrho=0$ implies that 
$h(\lambda)$ cannot vanish
on any open set, so we may replace $N$ by $\beta$, and the equations
to be solved are
\begin{align*}
0 & = \frac{m-2}{2} \left ( 1 + \beta \right ) 
+\frac{1}{R^{m-3}}
\delta_{\lambda} \left ( h R^{m-3} \right ),  \\
0 & = \frac{1}{2} \delta_{\lambda} \beta + \frac{h}{R} \delta_{\lambda}
\delta_{\lambda} R,
\end{align*}
which in terms of $\sigma:= \frac{h}{R}$ become
\begin{align}
0 & = \frac{m-2}{2} \left ( 1 + \beta \right ) 
+ R \delta_{\lambda} \sigma
+(m-2) \sigma \delta_{\lambda} R, \label{firsta} \\
0 & = \delta_{\lambda} \beta + 2 \sigma \delta_{\lambda}
\delta_{\lambda} R. \nonumber 
\end{align}
Taking $\delta_{\lambda}$ in the first equation and inserting the second yields
\begin{align*}
  (m-1) \delta_{\lambda} R \delta_{\lambda} \sigma + R \delta_{\lambda} \delta_{\lambda}
  \sigma =0 \quad \quad \Longleftrightarrow \quad \quad \delta_{\lambda}
  \left ( R^{m-1} \delta_{\lambda} \sigma \right ) =0.
\end{align*}
The quantity $R^{m-1} \delta_{\lambda} \sigma$ is of weight zero, so we can replace
$\delta_{\lambda}$ by $\frac{d}{d\lambda}$ and integrate
\begin{align}
  \delta_{\lambda} \sigma = \frac{-(m-2) M}{R^{m-1}}
\label{deltasigma}
\end{align}
where $M$ is an integration constant. Equation (\ref{firsta}) can be solved 
algebraically for $\beta$:
\begin{align*}
  \beta = -1 + \frac{2M}{R^{m-2}} - 2 \sigma \frac{dR}{d \lambda}
    \quad \quad \Longleftrightarrow
  \quad \quad N = \frac{\beta}{\sigma^2}
  = - \frac{1}{\sigma^2} \left ( 1 - \frac{2M}{R^{m-2}} + 2 \sigma
  \frac{d R}{d \lambda} \right ),
\end{align*}
where we have used that $R$ is of weight zero.
We emphasize that so far no choice
of  radial coordinate $\lambda$ has been made.
Concerning the tensor $Y$, in the gauge where $\ll=0$ it has the form
(see (\ref{nullgauge}))
\begin{align*}
  \Y  & = \Yr d \lambda^2 + \YT \gsph
  = \left ( \frac{d \ellr}{d \lambda} - \ellr \Gamma  \right ) d \lambda^2
  - \ellr R \sigma \gsph \\
  & = (\delta_{\lambda} \ellr) d \lambda^2 - \ellr R \sigma \gsph
\end{align*}
where in the second equality we used that $\ellr$ is of weight $1$.
Observe also that the gauge $\ll=0$ admits a residual gauge freedom
of the form $(u,0)$. This  may be used to set $\ellr =1$.
In this  gauge, the hypersurface data has two arbitrary
functions $R(\lambda)$ (weight $0$) and $\sigma(\lambda)$ (weight $-1$),
both non-vanishing.
Now, since $\sigma$ is of weight $-1$, (\ref{deltasigma}) reads
\begin{align*}
  \frac{d \sigma}{d \lambda} + \Gamma \sigma =
  - \frac{(m-2)M}{R^{m-1}} \quad \quad \Longleftrightarrow
  \quad \quad
  \Gamma =
  - \frac{1}{\sigma} \left ( \frac{d \sigma}{d\lambda}
  + \frac{(m-2) M}{R^{m-1}} \right )
\end{align*}
and the most general hypersurface data is
\begin{align*}
  \gamma & = \left [ - \frac{1}{\sigma^2} \left ( 1- \frac{2M}{R^{m-2}} \right )
    - \frac{2}{\sigma} \frac{dR}{d\lambda} \right ] d \lambda^2
  + R^2 \gsph, \quad \quad \ellr = 1, \quad \quad \ll=0 \\
  Y & = \frac{1}{\sigma} \left ( \frac{d \sigma}{d\lambda}
  + \frac{(m-2) M}{ R^{m-1}} \right ) d \lambda^2
  - R \sigma \gsph.
\end{align*}
We still have the freedom to
select the coordinate $\lambda$. Choosing  an affine
coordinate, where $\Gamma=0$ (cf. Lemma \ref{flat}), it holds
\begin{align*}
  \frac{d \sigma}{d \lambda} = - \frac{(m-2) M }{R^{m-1}},
\quad \quad  \quad \quad Y = - R \sigma  \gsph.
\end{align*}
So far, we have been able to identity the most general solution
of the spherically symmetric vacuum constraint equations in terms
of a constant $M$ and several free functions. 
Our final aim is to show that all this data can be embedded (locally)
in the Kruskal spacetime
of mass $M$ and dimension $m+1$. We work in an advanced 
Eddington-Finkelstein patch, where the spacetime metric takes the form
\begin{align*}
  g_{\mbox{\tiny Kr}}
  = - \left ( 1- \frac{2M}{r^{m-2}} \right ) dv^2 + 2 dv dr
  + r^2 \gsph.
\end{align*}
Consider the embedding $\Phi := \{ v(\lambda),r(\lambda), \omega^A\}$ defined by
\begin{equation*}
\left .   \begin{array}{ll}
  r(\lambda) &= R (\lambda) \\
  \frac{dv(\lambda)}{d\lambda} &= -\frac{1}{\sigma}
\end{array}
\right \}.
  \end{equation*}
It is immediate that the first fundamental form is $\gamma$. To check that
the data is embedded we need to find a rigging that yields
$\ellr, \ll$ and $\Y$.  This rigging must be null (to fulfill the gauge choice
$\ll=0$).
Moreover, it must be orthogonal to the spherical part $\gsph$ of the metric.
Select $\rig =  -\sigma \partial_r$. The
corresponding one form is $\bm{\rig} = -\sigma dv$ which satisfies
\begin{align*}
\Phi^{\star} (\bm{\rig} ) = -\sigma \frac{dv}{d\lambda} d \lambda
= d \lambda = \ellc
\end{align*}
given that $\ellc = \ellr d \lambda = d \lambda$. Finally, we compute
\begin{align*}
  \frac{1}{2}  \Phi^{\star} \left ( \pounds_{\rig}
  g_{\mbox{\tiny Kr}} \right )
  & = \frac{1}{2} \Phi^{\star} \left ( - \sigma
  \pounds_{\partial_r}
  g_{\mbox{\tiny Kr}}
  - 2 g_{\mbox{\tiny Kr}} (\partial_r, \cdot) \otimes_s
d \sigma
  \right ) \\
  & = - \frac{\sigma }{2} \Phi^{\star} \left (\pounds_{\partial_r}
  g_{\mbox{\tiny Kr}} \right ) + \frac{1}{\sigma} \Phi^{\star} (\bm{\rig})
  \otimes_s d \sigma \\
  & = - \frac{\sigma }{2} \Phi^{\star} \left (
- \frac{ 2M (m-2)}{r^{m-1}} dv^2 + 
  2 r \gsph \right ) + \frac{1}{\sigma}
  \frac{d \sigma}{d \lambda}  d \lambda^2 \\
  & = \left ( \frac{ M (m-2)}{\sigma R^{m-1}} + \frac{1}{\sigma}
  \frac{d \sigma}{d\lambda} \right ) d \lambda^2 - \sigma R \gsph
  = \Y.
\end{align*}
Summarizing, we have proved
\begin{theorem}[Birkhoff theorem for hypersurface data]
\label{Birkhoff}
  Let $\hypdata$ be spherically symmetric and vacuum  hypersurface data
  of dimension $m > 2$ and of Lorentzian ambient
  signature. There exists $M \in \mathbb{R}$ such that
  outside a (possibly empty) closed
  subset of $\N$ with empty interior, there exist local coordinates
  $\{\lambda,\omega^A\}$,
 two smooth, positive functions
 $R(\lambda)$ and $\sigma(\lambda)$  and a choice of gauge such that
   \begin{align*}
\  \gamma & = \left [ - \frac{1}{\sigma^2} \left ( 1- \frac{2M}{R^{m-2}} \right )
    - \frac{2}{\sigma} \frac{dR}{d\lambda} \right ] d \lambda^2
  + R^2 \gsph, \quad \quad \quad \ellc = d \lambda, \quad \quad \quad \ll=0, \\
  Y & = \frac{1}{\sigma} \left ( \frac{d \sigma}{d\lambda}
  + \frac{(m-2) M}{ R^{m-1}} \right ) d \lambda^2
  - R \sigma \gsph.
   \end{align*}
   The functions $R$ and $\sigma$ are, respectively, of weight $0$ and
   $-1$. Moreover, all such  hypersurface data can be embedded into the Kruskal
   spacetime of mass $M$.
\end{theorem}
\begin{remark}
  The closed set in the statement of the theorem corresponds to the
  centers of symmetry of the $SO(m)$ action and to points
  where $\sigma$ might vanish.
\end{remark}
\begin{remark}
  This theorem is more general that the standard Birkhoff theorem
  because the hypersurface data constraint equations may
  a priori have more solutions than those coming from embedded hypersurfaces
  in vacuum spacetimes. If the tensor $\gamma$ is positive definite, then
  the data defines a vacuum initial data set for the Einstein vacuum field
  equations, and well-posedness of the Cauchy problem of the
  Einstein field equations
  combined with the standard Birkhoff theorem would provide an
  alternative proof of the hypersurface
  data Birkhoff theorem \ref{Birkhoff} in this case.
  If the data is null everywhere (i.e.
  $\gamma$ is degenerate everywhere), one could still argue via
  well-posedness of the characteristic
  initial value (this requires constructing out of the
  prescribed hypersurface data, data on two
  null hypersurfaces intersecting on a spacelike surface).
  However, for any other case
  ($\gamma$ of Lorentzian signature, or of changing signature), no
  such well-posedness result is known -- neither expected, for instance
when $\gamma$ is Lorentzian somewhere --) so one cannot argue indirectly
  via the Cauchy problem. In this sense, theorem \ref{Birkhoff} is a
  non-trivial generalization of the classical Birkhoff theorem.
\end{remark}

\section*{Acknowledgments}

Financial support under the projects
PGC2018-096038-B-I00
(Spanish Ministerio de Ciencia, Innovaci\'on y Universidades)
and SA083P17 (JCyL)
is acknowledged. I am grateful to Miguel Manzano for checking the paper and
spotting several typos.

\end{document}